\documentclass[a4paper,10pt,twoside,english]{scrartcl}%{scrreprt}
\usepackage{marginnote}
\usepackage[utf8x]{inputenc}
\usepackage[unicode=true]{hyperref}
\usepackage{amsmath,amssymb,amsthm}
\usepackage{mathrsfs}  
\usepackage[all]{hypcap}
\usepackage[T1]{fontenc}
\usepackage{lmodern}
\usepackage{graphicx}
\usepackage{enumerate}
\usepackage{setspace} % line spacing

\usepackage{blkarray}

\newcommand{\titel}{The Local Semicircle Law for \\Random Matrices with a Fourfold Symmetry}
\title{\titel} 
\author{Johannes Alt\footnote{Partially funded by ERC Advanced Grant RANMAT No. 338804.}\\{\small IST Austria, Am Campus 1, A-3400 Klosterneuburg, Austria. jalt@ist.ac.at}}% \and László Erd\H{o}s}
\date{}

\hypersetup{
	pdftitle={\titel},
	pdfauthor={Johannes Alt}%,
}

\numberwithin{equation}{section}
\newcommand{\R}{\mathbb{R}}  % The real numbers.
\ifdefined\C\renewcommand{\C}{\mathbb{C}}\else\newcommand{\C}{\mathbb{C}}\fi % Complex numbers
\renewcommand{\Im}{\mathrm{Im}\,} %imaginary part of a complex number
\renewcommand{\i}{\mathrm{i}\,} %imaginary part of a complex number
\newcommand{\N}{\mathbb{N}}  % Positive integers.	
\newcommand{\E}{\mathbb{E}}  % expected value of random variable	
\newcommand{\e}{\mathcal{E}}  % symbol for error terms
\newcommand{\F}{\mathbb{F}}  % X - E X	
\newcommand{\Z}{\mathbb{Z}}  % Integers
\newcommand{\di}{\text{d}} % differential
\newcommand{\eps}{\varepsilon} % "correct" epsilon
\newcommand*{\defeq}{\mathrel{\vcenter{\baselineskip0.5ex \lineskiplimit0pt\hbox{\scriptsize.}\hbox{\scriptsize.}}}=}
                     
\providecommand{\norm}[1]{\lVert#1\rVert} %norm
\providecommand{\abs}[1]{\lvert#1\rvert} %absolute value

\begin{document}
%for Definitions, Propositions and etc.
\newtheoremstyle{test}% name
  {}%      Space above, empty = `usual value'
  {}%      Space below
  {\itshape}% Body font
  {}%         Indent amount (empty = no indent, \parindent = para indent)
  {\bfseries}% Thm head font
  {.}%        Punctuation after thm head 
  { }% Space after thm head: \newline = linebreak
  {}%         Thm head spec

\theoremstyle{test}
\newtheorem{defi}{Definition}[section]
\newtheorem{ex}[defi]{Example}
\newtheorem{rem}[defi]{Remark}
\newtheorem{exer}[defi]{Exercise}
\newtheorem{thm}[defi]{Theorem}
\newtheorem{lem}[defi]{Lemma}
\newtheorem{cor}[defi]{Corollary}
\newtheorem{pro}[defi]{Proposition}
\newtheorem*{rem*}{Remark}   %no numbering
\newtheorem*{ex*}{Example}   %no numbering
\newtheorem*{pro*}{Proposition} %no numbering
\newtheorem*{def*}{Definition}
\newtheorem*{cor*}{Corollary}
\newtheorem*{thm*}{Theorem}

\maketitle
\vspace*{-1.1cm}
\begin{abstract}
We consider real symmetric and complex Hermitian random matrices with the additional symmetry $h_{xy}=h_{N-y,N-x}$. 
The matrix elements are independent (up to the fourfold symmetry) and not necessarily
identically distributed. This ensemble naturally arises as the Fourier transform of a Gaussian orthogonal ensemble (GOE). 
It also occurs as the flip matrix model -- an approximation of the two-dimensional Anderson model at small disorder. 
We show that the density of states converges to the Wigner semicircle law despite the new symmetry type.
We also prove the local version of the semicircle law on the optimal scale.
\end{abstract}


\section{Introduction}

In 1955, Wigner conjectured that the eigenvalues of large random matrices describe the energy levels of large atoms \cite{Wigner1955}. 
Therefore, the distribution of the eigenvalues of a random matrix is an interesting and often studied object in random matrix theory. 
For an $N\times N$ random matrix with eigenvalues $(\lambda_i)_{i=1}^N$, let $\mu_N \defeq N^{-1} \sum_{i=1}^N \delta_{\lambda_i}$
be the \emph{empirical spectral measure}.
The celebrated Wigner semicircle law \cite{Wigner1955} asserts that $\mu_N$ converges to the semicircle law 
given by the density $\sqrt{(4-x^2)_+}/(2\pi)$ in the limit that the matrix size $N$ goes to infinity. 

The Wigner-Dyson-Gaudin-Mehta conjecture in \cite{mehta2004random} asserts 
that the distribution of the difference between consecutive eigenvalues of a large random matrix only depends on the symmetry type of the matrix and not 
on the distribution of the entries. This independence of the actual distribution is called universality. 
The proof of this conjecture by Erd\H os, Schlein, Yau and Yin in \cite{ErdosSchleinYau2011,erdoes_relaxation_flow_2012} 
is built upon establishing a local semicircle law in the first step (see \cite{ErdoesYau2012} for a review). 
An alternative approach was pursued by Tao and Vu in \cite{TaoVu2011_Acta}.

Wigner's semicircle law can be used to compute the number of eigenvalues contained in a fixed interval for a large random matrix.
With the help of a local semicircle law such prediction can also be made in the case of a variable interval size as long as it is considerably bigger than $N^{-1}$
 which is the typical distance of neighbouring eigenvalues. A local semicircle law is most commonly proved by establishing a convergence of the Stieltjes transform $m_N(z) \defeq 
N^{-1} \sum_{i=1}^N (\lambda_i - z)^{-1}$ of $\mu_N$ to the Stieltjes transform $m$ of Wigner's semicircle law. Then an interval size of $N^{-1}$ corresponds
to showing the convergence when $\eta=\Im z$ is of this order. 

One of the most general versions of a local semicircle law is presented in \cite{EJP2473}. They suppose that the random matrix $H=(h_{xy})_{x,y}$
is complex Hermitian (or real symmetric), i.e., $h_{xy}=\bar h_{yx}$ for all $x$ and $y$ with real-valued random variables $h_{xx}$ for all $x$ 
such that $(h_{xy})_{x \leq y}$ forms an independent family of centered random variables. 
Besides assuming that the variances $s_{xy}\defeq \E \abs{h_{xy}}^2$ of a row sum up to one, i.e,
\begin{equation}
\sum_{y} s_{xy} =1 
\label{eq:sum_s_xy_1_intro}
\end{equation}
for all $x$ which ensures that the eigenvalues stay of order 1, 
the most important requirement is the independence of the entries (up to the symmetry constraint).

Many works in random matrix theory start with this independence assumption. However, some naturally arising random 
matrix models do not fulfill it. An example is the Fourier transform of a Gaussian Orthogonal Ensemble (GOE).
For an $N\times N$ matrix $H=(h_{xy})_{x,y=1}^N$ the Fourier transform $\hat H=(\hat h_{pq})_{p,q \in \Z\slash N\Z}$ is defined through 
\[\hat h_{pq} = \frac{1}{N} \sum_{x,y=1}^N h_{xy}\exp\left( -\i \frac{2\pi}{N}(px-qy)\right) \]
for $p,q \in \Z\slash N\Z$. 
If $H=(h_{xy})_{x,y=1}^N$ is a real symmetric matrix then $\hat H=(\hat h_{pq})_{p,q \in \Z\slash N\Z}$ fulfills the relations 
\[\hat h_{pq}=\overline{\hat h}_{qp}=\hat h_{-q, -p}=\overline{\hat h}_{-p,-q}\]
for $p,q \in \Z\slash N\Z$. 
If the entries of $H$ are, in addition, centered Gaussian distributed random variables such that $\{h_{xy}; x \leq y\}$ are independent with $\E h_{xx}^2=2\E h_{xy}^2$ for $x\neq y$ 
then the entries of $\hat H$ will be independent up to this symmetry which we call \emph{fourfold symmetry}.

Interestingly, this symmetry also arises in random matrix approximations of the Anderson model. 
In \cite{Bellisard2003}, it is argued that the fourfold symmetry with a constant diagonal 
-- called the flip symmetry -- is a good approximation of the two-dimensional Anderson model in the regime of small disorder (see \cite{disertori2008} for a review on random matrix
models of the Anderson model).

The first local law for Wigner matrices on the optimal scale $\eta \approx N^{-1}$ (with logarithmic corrections) in the bulk has been proved by 
Erd\H os, Schlein and Yau in \cite{localsemicirclelaw1}. 
In \cite{localsemicircleBernoulli}, Erd\H os, Yau and Yin proved that $m_N-m$ is of the optimal order $(N\eta)^{-1}$ in the bulk and they could extend this result to the 
edges in \cite{Erdoes20121435}. In the more general case with non-identical variances and the assumption \eqref{eq:sum_s_xy_1_intro}, a local semicircle law on the scale 
$\eta \approx M^{-1}$ with $M\defeq (\max_{x,y} s_{xy})^{-1}$ has been established by Erd\H os, Yau and Yin in \cite{Bulk_Universality_Gen_Wigner_Matrices}.
For this case, Erd\H os, Knowles, Yau and Yin obtained the optimal order $(M\eta)^{-1}$ of $m_N-m$ in \cite{EJP2473} even at the edge.
A more detailed overview of the historical development of the local semicircle law can be found in section 2.1 of \cite{2012arXiv1212.0839E}. 

Our main result is a proof of the local semicircle law for random matrices possessing 
the fourfold symmetry. Despite the different symmetry type compared to the case in \cite{EJP2473} the limiting distribution of the 
empirical spectral measure will still be Wigner's semicircle law. The basic structure of the proof follows \cite{EJP2473}. The main novelty is that not only the diagonal elements of the Green function have to be 
treated separately from the offdiagonal ones, but elements on the counterdiagonal need to be estimated separately via a new 
self-consistent equation.

We conclude this introduction with an outline of the structure of the present article. In the following section, we 
introduce our model and some notation and state our main result. In section 3, we prove that the Fourier transform of a GOE 
satisfies the assumptions of Theorem \ref{thm:Main_Fourfold}. The remaining part is devoted to the proof 
of our main result. Section 4 contains a collection of the tools used in the proof which is given in the subsequent section.  
In the appendix, we show that the fluctuation averaging holds true for the fourfold symmetry as well.

\emph{Acknowledgement:} I am very grateful to László Erd\H os for drawing my attention to this question, for suggesting the method and for numerous helpful 
comments during the preparation of this article. 
Moreover, I thank Oskari Ajanki and Torben Krüger for useful discussions. 

\section{Main Result}

For $N\in \N$ and $x, y \in \Z\slash N\Z$, let $\zeta_{xy}^{(N)}$ be real or complex valued random variables (in the following we drop the $N$-dependence in our notation) such that 
$\zeta_{xx}$ is real valued, $\E \zeta_{xy}=0$ and $\E \abs{\zeta_{xy}}^2=1$ for all $x, y$. 
Moreover, we assume that for every $p \in \N$ there is a constant $\mu_p$ such that 
\begin{equation}
\E\abs{\zeta_{xy}}^p\leq \mu_p
\label{eq:finite_moments}
\end{equation}
for all $x,y\in \Z\slash N\Z$ and $N\in \N$.
For fixed $N\in \N$, the entries are supposed to be independent up to the fourfold symmetry  
$\zeta_{xy}=\bar \zeta_{yx}=\zeta_{-y, -x}=\bar \zeta_{-x,-y}$
%\label{eq:fourfold_symmetry}
%\end{equation}
for all $x,y \in \Z\slash N\Z$.

For $N \in \N$, let $S=(s_{xy})_{x,y\in \Z\slash N\Z}$ be an $N\times N$-matrix of nonnegative real numbers such that 
%\begin{equation*}
$s_{xy}=s_{yx}=s_{-y, -x}=s_{-x,-y}$
%\end{equation*}
for all $x,y$ and $S$ is stochastic, i.e., for every $x$ we have 
\begin{equation}
\sum_y s_{xy}=1.
\label{eq:condition_s_xy} 
\end{equation}
Furthermore, we assume that the $N$-dependent parameter $M\defeq (\max_{x,y} s_{xy})^{-1}$ satisfies 
\begin{equation}
N^\delta\leq M \leq N
\label{eq:assumption_M}
\end{equation}
for some $\delta>0$. Note that the first estimate is an assumption on $S$ whereas the second bound follows 
from the definition of $M$ and \eqref{eq:condition_s_xy}.

Defining $h_{xy}\defeq s_{xy}^{1/2}\zeta_{xy}$ we obtain the Hermitian random matrix $H^{(N)}=(h_{xy})_{x,y\in \Z\slash N\Z}$ which 
fulfills the following fourfold symmetry
\begin{equation}
h_{xy}=\bar h_{yx}=h_{-y, -x}=\bar h_{-x,-y}
\label{eq:fourfold_symmetry}
\end{equation}
because of the definition of $\zeta_{xy}$ and the conditions on $S$.
%\eqref{eq:fourfold_symmetry} as well. 
By definition, $S$ describes the variances of $H^{(N)}$.

%\subsection{Main Result}
Let $\rho$ denote Wigner's semicircle law and $m$ its Stieltjes transform, i.e.,
\begin{equation}
\rho(x)\defeq \frac{1}{2\pi} \sqrt{(4-x^2)_+}, \qquad m(z)\defeq\frac{1}{2\pi}\int_{-2}^{2} \frac{\sqrt{4-x^2}}{x-z} \di x
\label{eq:semicircle1}
\end{equation}
for $x\in \R$ and $z \in \C\backslash \R$. 
For the real and imaginary part of $z\in \C$, we will use the abbreviations $E$ and $\eta$, respectively, i.e., $z=E+\i \eta$ with $E, \eta \in \R$.

With this definition the complex valued function $m(z)$ is the unique solution of 
\begin{equation}
m(z)+\frac{1}{m(z)+z}=0
\label{eq:semicircle2}
\end{equation}
such that $\Im m(z)>0$ for $\eta >0$. 
Denoting the resolvent or Green function of $H$ by
\[ G(z) \defeq (H-z)^{-1}\]
and its entries by $G_{ij}(z)$ for $z\in \C\backslash \R$ we obtain for the Stieltjes transform $m_N$ of the empirical spectral measure
\[m_N(z) = \frac{1}{N} \text{Tr} G(z).\]

%\subsection{Stochastic Domination and Spectral Domains}

%Stochastic Domination and Spectral domains
We use the definitions of stochastic domination and spectral domain given in \cite{EJP2473}.

\begin{defi}[Stochastic Domination]
Let $X=(X^{(N)}(u); u \in U^{(N)}, N \in \N)$ and $Y=(Y^{(N)}(u); u \in U^{(N)}, N \in \N)$ be two 
families of nonnegative random variables for a possibly $N$-dependent parameter set $U^{(N)}$. 
We say that $X$ is \emph{stochastically dominated} by $Y$, uniformly in $u$, if for all $\eps>0$
and $D>0$ there is a $N_0(\eps, D)\in \N$ such that
$$\sup_{u\in U^{(N)}}\mathbb P\left[X^{(N)}(u) > N^\eps Y^{(N)}(u)\right]\leq N^{-D}$$
for all $N \geq N_0$. In this case, we use the notation $X\prec Y$. 
If $X$ is a family consisting of complex valued random variables and $\abs{X} \prec Y$ then we 
write $X\in O_\prec (Y)$.
\end{defi}

The definition of stochastic domination implies the following estimate which is important for our arguments
\begin{equation}
\abs{h_{xy}}\prec s_{xy}^{1/2} \leq M^{-1/2}.
\label{eq:estimate_h_xy_stoch_dom}
\end{equation}

\begin{defi}
An $N$-dependent family $\mathbf D = (\mathbf D^{(N)})_{N \in \N}$ of subsets of the complex plane with
$$ \mathbf D^{(N)} \subset \{z=E+\i\eta \in \C; E\in [-10,10], M^{-1}\leq \eta \leq 10\}$$
for every $N \in \N$ is called a \emph{spectral domain}.
\end{defi}

In analogy to the matrix $S$, we define $R=(r_{xy})=(\E h_{xy}^2)^{x\neq -x}_{y\neq -y}$. If $N$ is odd then $R$ is 
an $(N-1)\times(N-1)$ matrix, otherwise it is an $(N-2)\times(N-2)$ matrix.
%$r_{xy} \defeq \E h_{xy}^2$ for $x\neq -x$ and $y\neq -y$ and $R\defeq (r_{xy})^{x\neq -x}_{y\neq -y}$ 
For $\eta>0$, we introduce the corresponding two control parameters
\begin{equation}
\Gamma_S(z) \defeq \norm{(1-m^2(z)S)^{-1}}_{\ell^\infty\to \ell^\infty},\qquad \Gamma_R(z)  \defeq \norm{(1-m^2(z)R)^{-1}}_{\ell^\infty\to \ell^\infty}
\label{eq:definition_Gamma}
\end{equation}
and their maximum $\Gamma(z) \defeq \max\{\Gamma_S(z),\Gamma_R(z)\}$ (Note that $\Gamma_S$ is denoted by $\Gamma$ in \cite{EJP2473}).
%Compared to \cite{EJP2473} the parameter $\Gamma_R$ is new, whereas $\Gamma_S$ was present there and denoted by $\Gamma$.
%The new fourfold symmetry requires a second self-consistent equation for the proof of the local semicircle law. For this self-consistent equation 
%which describes the behaviour of the offdiagonal terms $G_{x,-x}$ 
%the parameter 

For the definition of the spectral domain underlying our estimates, we define
\begin{equation}
\eta_E\defeq\min\left\{\eta; \frac{1}{M\eta}\leq \min\left\{\frac{M^{-\gamma}}{\Gamma(z)^3},\frac{M^{-2\gamma}}{\Gamma(z)^4 \Im m(z)}\right\} 
\text{ for all }z\in [E+\i\eta ,E+\i 10]\right\}
\label{eq:definition_eta_E}
\end{equation}
for $\gamma\in (0,1/2)$ and $E\in \R$.
Then, for $\gamma \in (0,1/2)$ the spectral domain $\mathbf S\equiv \mathbf S(\gamma)= (\mathbf S^{(N)})_{N \in \N}$ is defined as
\begin{equation}
\mathbf S^{(N)} \defeq \left\{E+\i\eta; \abs{E}\leq 10, \eta_E\leq \eta \leq 10\right\}.
%remark at some point that spectral domain doesn't have to be changed if the off-diagonal terms of R are zero.
\label{eq:definition_S}
\end{equation}
Note that the spectral domain $\mathbf S$ differs from the spectral domain $\mathbf S$ in \cite{EJP2473} due to the new definition of $\Gamma(z)$.
Besides this difference the following main result of this article has the same form as Theorem 5.1 in \cite{EJP2473}.

\begin{thm}[Local Semicircle Law]
\label{thm:Main_Fourfold}
Let $H$ be a random matrix with the fourfold symmetry \eqref{eq:fourfold_symmetry} such that the conditions \eqref{eq:finite_moments} and \eqref{eq:condition_s_xy} are fulfilled.
For $\gamma \in (0,1/2)$, we have 
\begin{equation}
%\abs{G_{ij}(z)-\delta_{ij}m(z)}\prec \Pi(z)\defeq \sqrt{\frac{\Im m(z)}{M\eta}} + \frac{1}{M\eta}
\abs{G_{xy}(z)-\delta_{xy}m(z)}\prec \sqrt{\frac{\Im m(z)}{M\eta}} + \frac{1}{M\eta}
\label{eq:main_fourfold_estimate_Lambda}
\end{equation}
uniformly in $x,y$ and $z \in \mathbf S$, as well as 
\begin{equation}
\abs{m_N(z)-m(z)} \prec \frac{1}{M\eta}
\label{eq:main_fourfold_estimate_m_N}
\end{equation}
uniformly in $z \in \mathbf S$.
\end{thm}

The proof of our main result is based on studying self-consistent equations in the same way as the proof of Theorem 5.1 in \cite{EJP2473} 
which uses one self-consistent equation for $G_{xx}-m$. 
However, due to the fourfold symmetry it is no longer possible to directly show that the entries $G_{x,-x}$ are small as in \cite{EJP2473}.
Therefore, we introduce a second, new self-consistent equation for $G_{x,-x}$.
While deriving these self-consistent equations we will see that the expressions $G_{xx}-m$ for $x \in \Z\slash N\Z$ and $G_{x,-x}$ for $x\neq -x$ are
connected among each other via $\E \abs{h_{xa}}^2$ and $\E h_{xa}^2$, respectively. Therefore, we introduce the matrix $R$ in an analogous fashion as $S$ is introduced in \cite{EJP2473}. 
The corresponding control parameters $\Gamma_R$ and $\Gamma_S$ will appear in our estimates in section \ref{sec:preliminary_bound}. 
Whereas the latter control parameter is present in \cite{EJP2473} and denoted by $\Gamma$ in there, the matrix $R$ and the corresponding parameter $\Gamma_R$ are new in our work. 
The role of $\Gamma$ in \cite{EJP2473} is filled by the maximum $\Gamma(z) = \max\{\Gamma_S(z),\Gamma_R(z)\}$. 
Estimates on $\Gamma$ similar to the ones in \cite{EJP2473} are collected in Lemma \ref{lem:estimate_Gamma} and Remark \ref{rem:estimate_Gamma_R}.

\begin{rem}\label{rem:smaller_spectral_domain}
%In case of real valued random variables we have $\Gamma_R = \Gamma_S$. 
%\marginnote{Estimates if $h_{xy}$ is real}
If the random variables $h_{xy}$ are complex valued with $\E h_{xy}^2=0$ for all $x\neq y$ then $\Gamma_R(z) \leq C\Gamma_S(z)$ for $z \in \{ E+\i \eta; E \in [-10,10], \eta \in (0,10]\}$ and therefore 
we can replace $\Gamma$ by $\Gamma_S$ in \eqref{eq:definition_eta_E}. 
Thus, in this case, our estimates hold on the spectral domain used in Theorem 5.1 in \cite{EJP2473}.
\end{rem}

To have a shorter notation in the following arguments, we introduce the $z$-dependent stochastic control parameters
\begin{align}
 \Lambda_d(z)\defeq &\max_x \abs{G_{xx}(z)-m(z)},&   \Lambda_g(z)\defeq & \max_{x\neq y\neq -x}\abs{G_{xy}(z)}, & \Lambda_-(z)\defeq & \max_{x\neq -x}\abs{G_{x,-x}(z)}, \nonumber \\
\Lambda_o(z)\defeq &\max\{\Lambda_g(z),\Lambda_-(z)\}, & \Lambda(z)\defeq & \max\{\Lambda_d(z),\Lambda_o(z)\}.
\label{eq:definition_Lambda}
\end{align}
Compared to \cite{EJP2473} we added the control parameter $\Lambda_-$ since the off-diagonal terms $G_{x,-x}$ will be estimated differently than the generic
off-diagonal terms. 

\section{Fourier Transform of Random Matrices} \label{ch:fouriertransform}

%\begin{defi}[Gaussian Ensembles]
%Let $(X_i)_i$, $(Y_{ij})_{i<j}$ and $(Z_{ij})_{i<j}$ be independent families of $N(0,1)$-distributed independent random variables.  
%\begin{enumerate}[(i)]
%\item We consider the (symmetric) $N\times N$ matrix $H_N=(h^{(N)}_{ij})_{i,j}$ with
%\begin{align*}
%h^{(N)}_{ij}=h^{(N)}_{ji}=&\frac{1}{\sqrt{N}} Y_{ij} \qquad \text{ for }  i<j, \\ 
%h^{(N)}_{jj}=&\frac{1}{\sqrt{N}}\sqrt{2} X_j 
%\end{align*}
%for $N \in \N$. Then $H_N$ is said to be an element of GOE$_N$ and the sequence $(H_N)_{N\in \N}$ is called a \emph{Gaussian orthogonal ensemble}. 
%\item For $N \in \N$ we define the (Hermitian) $N \times N$ matrix $H_N=(h^{(N)}_{ij})_{i,j}$ by
%\begin{align*}
%h^{(N)}_{ij}=\overline{h^{(N)}_{ji}}=&\frac{1}{\sqrt{2N}}(Y_{ij}+iZ_{ij}) \qquad \text{ for }  i<j, \\ 
%h^{(N)}_{jj}=&\frac{1}{\sqrt{N}} X_j. 
%\end{align*}
%The matrix $H_N$ is said to belong to GUE$_N$ and the sequence $(H_N)_{N\in \N}$ is called a \emph{Gaussian unitary ensemble}.
%\end{enumerate}
%\end{defi}

%The adjective ``orthogonal'' (``unitary'' respectively) in the previous definition comes from the fact that conjugating $H_N$ with an orthogonal (unitary)
%$N\times N$ matrix $O$ yields again an element in GOE$_N$ (GUE$_N$), i.e. the joint distribution of the entries of $H_N$ is invariant under 
%conjugation with an orthogonal (unitary) matrix. 

In this section, we give an example of a random matrix satisfying the conditions of Theorem \ref{thm:Main_Fourfold}, namely the Fourier transform (in the following sense) of a Gaussian orthogonal ensemble. 

\begin{defi}[Fourier Transform]
Let $H=(h_{xy})_{x,y=1}^N$ be an $N\times N$ matrix. The \emph{Fourier transform} $\hat{H}=(\hat h_{pq})_{p,q \in \Z\slash N\Z}$ is the $N\times N$ matrix whose entries are given by 
$$ \hat{h}_{pq}=\frac{1}{N}\sum_{x,y=1}^N h_{xy} \exp\left(-\i\frac{2\pi}{N}(px-qy)\right)$$ 
for $p,q \in \Z\slash N\Z$. 
\end{defi}

In the next Lemma we collect the basic properties of the Fourier transform of a Gaussian orthogonal ensemble which will imply the conditions of Theorem \ref{thm:Main_Fourfold}.
%compare to the result for a GUE, reason for that Eh_{x_1y_1}h_{x_2y_2} = N^{-1}\delta_{x_1y_2}\delta_{y_1x_2} for GUE matrix
%
%%%%%%%%   IMPORTANT REMARK ABOUT INDEX SET AND WELL-DEFINENESS $$$$ 
%% We consider the matrix $\hat{H}$ to be indexed by elements in $(\Z\slash N\Z)\times (\Z\slash N\Z)$. 
%% Note that $\exp(-i \frac{2\pi}{N} pk)$ is well-defined for $p\in \Z\slash N\Z$ and $k\in \Z$.
%%%%%%%%%%%%%%%%%%%%%%%

%This means $-[p]=[N-p]$ where $[m]$ denotes the equivalence class of $m \in \Z$. In the following we drop the square brackets for elements in $\Z\slash N\Z$. 
% since $\exp(-i2\pi m)=1$ for $m \in \Z$.

%\newpage 
\begin{lem}
Let $H$ be a GOE and $\hat{H}$ its Fourier transform. Then the entries
$\hat{h}_{pq}$ and $\hat{h}_{rs}$ are independent if and only if $$(p,q)\notin \{(r,s),(s,r),(-r,-s),(-s,-r)\}.$$
Moreover, $\hat{H}$ satisfies the fourfold symmetry \eqref{eq:fourfold_symmetry} for all $p,q\in \Z\slash N\Z$.
%we have $$\hat{H}_{pq}=\overline{\hat{H}_{qp}}=\hat{H}_{-q,-p}=\overline{\hat{H}_{-p,-q}}$$ for all $p, q$. 
We have 
\begin{equation}
\E \abs{\hat{h}_{pq}}^2=N^{-1}, \quad \E \hat{h}_{pr}^2 =0
\label{eq:fourier_properties}
\end{equation}
for all $q$ and $p\neq r$.
\end{lem}

\begin{proof}
To prove the if-part it suffices to show that $\hat{H}$ satisfies \eqref{eq:fourfold_symmetry} which is a direct consequence of the fact that $H$ is symmetric.

Since $\hat{h}_{pq}$ and $\hat{h}_{rs}$ are jointly normally distributed and $\E \hat{h}_{pq}=\E\hat{h}_{rs}=0$, it suffices to prove that 
$\E\hat{h}_{pq}\overline{\hat{h}_{rs}}=0$ and $\E\hat{h}_{pq}\hat{h}_{rs}=0$ in order to show that these random variables are independent.
The formula $\E h_{x_1y_1}h_{x_2y_2}=N^{-1}(\delta_{x_1x_2}\delta_{y_1y_2}+\delta_{x_1y_2}\delta_{y_1x_2})$ together with 
$$\sum_{x=1}^{N} \exp\left(-\i\frac{2\pi}{N}mx\right)=\begin{cases} N, & m=0,\\ 0, &\text{otherwise}\end{cases}$$
for $m \in \Z\slash N\Z$ yields
$\E \hat{h}_{pq}{\hat{h}_{rs}}=N^{-1}$ for $(p,q)\in \{(s,r),(-r,-s)\}$ and $\E \hat{h}_{pq}\hat{h}_{rs}= 0$ otherwise.
Thus, $\E \hat{h}_{pq} {\hat{h}_{rs}}\neq 0$ if and only if $(p,q)\in \{(s,r),(-r,-s)\}$. In particular, $\E\hat{h}_{pq}^2=0$ for $p \neq q$.

The relation $\overline{\hat{h}_{rs}} = {\hat{h}_{sr}}$ implies the first part of \eqref{eq:fourier_properties} and concludes the proof of 
the only-if part.
\end{proof}

Therefore, the Fourier transform of a Gaussian orthogonal ensemble fulfills all requirements of  %the local semicircle law in the case of the fourfold symmetry, 
Theorem \ref{thm:Main_Fourfold}  with $s_{pq}\defeq N^{-1}$ and $\zeta_{pq}\defeq N^{-1/2}\hat{h}_{pq}$.
Because of the first result in \eqref{eq:fourier_properties} the condition \eqref{eq:condition_s_xy} is fulfilled.
By the second part of \eqref{eq:fourier_properties} Remark \ref{rem:smaller_spectral_domain} is applicable.
Thus, the local semicircle law holds true for these random matrices.

\section{Tools}
In this section, we collect the tools for the proof of Theorem \ref{thm:Main_Fourfold}. We start with listing some resolvent identities which are the basic tool for all our estimates
as they encode the dependences between diagonal and off-diagonal entries of the resolvents. Computing the partial expectation of certain terms in expansions of the resolvent entries 
with respect to a minor will be an important step to derive the self-consistent equations. Thus, we introduce some notation in the second subsection. We conclude with the fluctuation averaging, 
an important mechanism to improve some bounds, and some estimates on $m$ and $\Gamma$ which are frequently used in our proofs.

\subsection{Minors and Resolvent Identities}
Let $H=(h_{xy})_{x,y \in \Z\slash N\Z}$ be a Hermitian matrix and $\mathbb T \subset \Z\slash N\Z$.  

\begin{defi} \label{def:minors_()} 
We define the $N\times N$ matrix $H^{(\mathbb T)}$ and its \emph{resolvent} or \emph{Green function} $G^{(\mathbb T)}$ through
$$(H^{(\mathbb T)})_{ij} \defeq \mathbf 1(i \notin \mathbb T)\mathbf 1(j \notin \mathbb T) h_{ij}, \quad G^{(\mathbb T)}(z) \defeq (H^{(\mathbb T)}-z)^{-1}$$
for $i,j \in \Z\slash N\Z$ and for $z\in \C\backslash\R$. We denote the entries of $G^{(\mathbb T)}(z)$ by $G_{ij}^{(\mathbb T)}(z)$.
%denote the resolvent of $H^{(\mathbb T)}$ in $z$.
We set 
$$\sum_i^{(\mathbb T)} \defeq \sum_{i; i \notin \mathbb T}.$$
In both cases, we write $(a_1,\ldots,a_n,\mathbb T)$ for $(\{a_1,\ldots,a_n\}\cup \mathbb T)$.
\end{defi}

Note that $H^{(\mathbb T)}$ is still a Hermitian $N \times N$ matrix, in particular $G^{(\mathbb T)}$ exists.
To estimate the resolvent entries we make essential use of the following relations.

\begin{lem}[Resolvent Identities] \label{Lem:resolv_ident} 
%For a Hermitian matrix $H=(h_{ij})_{i,j=-N/2}^{N/2}$ and \\$\mathbb T \subset \{-N/2,\ldots,N/2\}$
For $i,j,k \notin \mathbb T$, the following statements hold:\\
%If $i\notin \mathbb T$, Schur's complement formula holds
\begin{equation}
\frac{1}{G^{(\mathbb T)}_{ii}}=h_{ii}-z-\sum_{a,b}^{(\mathbb T,i)}h_{ia}G_{ab}^{(\mathbb T,i)}h_{bi}.
\label{eq:Schur_formula}
\end{equation}
%If $i,j,k \notin \mathbb T$ and $i,j \neq k$ then
If $i,j \neq k$ then
\begin{equation}
G_{ij}^{(\mathbb T)} = G_{ij}^{(\mathbb T,k)}+\frac{G_{ik}^{(\mathbb T)}G_{kj}^{(\mathbb T)}}{G_{kk}^{(\mathbb T)}},\quad
\frac{1}{G_{ii}^{(\mathbb T)}}=\frac{1}{G_{ii}^{(\mathbb T,k)}} -\frac{G_{ik}^{(\mathbb T)}G_{ki}^{(\mathbb T)}}{G_{ii}^{(\mathbb T)}G_{ii}^{(\mathbb T,k)}G_{kk}^{(\mathbb T)}}.
\label{eq:resolvent_identity1}
\end{equation}
%If $i,j \notin \mathbb T$ satisfy $i \neq j$ then 
If $i \neq j$ then 
\begin{equation}
G_{ij}^{(\mathbb T)}=-G_{ii}^{(\mathbb T)}\sum_a^{(\mathbb T,i)} h_{ia}G_{aj}^{(\mathbb T,i)}=-G_{jj}^{(\mathbb T)}\sum_a^{(\mathbb T,j)}G_{ia}^{(\mathbb T,j)}h_{aj}.
\label{eq:resolvent_identity2}
\end{equation}
\end{lem}

The proof of Schur's complement formula, \eqref{eq:Schur_formula}, and the first identity in \eqref{eq:resolvent_identity1} can be found in Lemma 4.2 in \cite{Bulk_Universality_Gen_Wigner_Matrices}
and the second identity follows directly from the first one. Lemma 6.10 in \cite{Spec_Stat} contains a proof of \eqref{eq:resolvent_identity2}.

Moreover, if $\eta>0$ then the spectral theorem for self-adjoint matrices yields
\begin{equation}
\sum_l \abs{G_{kl}^{(\mathbb T)} (z)}^2 = \frac 1 \eta \Im G_{kk}^{(\mathbb T)} (z).
\label{eq:Ward_identity}
\end{equation}
This identity is sometimes called \emph{Ward identity}. 

The functional calculus implies the following estimates on the entries of the resolvent:
\begin{equation}
\abs{G_{ij}^{(\mathbb T)}(z)} \leq \eta^{-1} \leq M
\label{eq:deterministic_bound_G_ij^T}
\end{equation}
for $\eta>0$ and all $i,j\in \Z\slash N\Z$. The second estimate holds if $z\in \mathbf D$ where $\mathbf D$ is a spectral domain.

\subsection{Partial Expectation}

For the partial expectation with respect to the $\sigma$-algebra generated by $H^{(x,-x)}$, we introduce the following notation.
\begin{defi}[Partial Expectation] \label{def:partial_expectation}
Let $X$ be an integrable random variable. For $x \in \Z\slash N\Z$ we define the random variables $\E_x X$ and $\F_x X$ through
$$ \E_x X \defeq \E[X|H^{(x,-x)}], \quad \F_x X \defeq X -\E_x X.$$
The random variable $\E_x X$ is called the \emph{partial expectation} of $X$ with respect to $x$. 
\end{defi}

The symbols $\E_x$ and $\F_x$ are the analogues of $P_i$ and $Q_i$ in \cite{EJP2473} that were defined by considering the minor $H^{(i)}$. 
Due to the fourfold symmetry column $x$, $-x$ and row $x$, $-x$ contain the same information, so the conditional expectation is taken with 
respect to the minor $H^{(x,-x)}$. Notice that it may happen that $x=-x$, in which case $H^{(x,-x)}$ is an $(N-1)\times(N-1)$ minor. 

\begin{defi}[Independence] \label{def:independent}
We say that the integrable random variable $X$ is \emph{independent} of $\mathbb T\subset \Z\slash N\Z$ if $X=\E_x X$ for 
all $x \in \mathbb T$.
\end{defi}

\noindent If $Y$ is independent of $x$ then $\F_x(X)Y = XY - \E_x(X\E_x Y)=\F_x(XY)$ and therefore 
\begin{equation}
\E \F_x(X)Y = \E \F_x (XY) = \E (XY) -\E \E_x (XY)=0.
\label{eq:independent_expectation}
\end{equation}

\subsection{Fluctuation Averaging}

Let $\mathbf D$ be a spectral domain, $H$ satisfy the requirements of Theorem \ref{thm:Main_Fourfold} and $\Psi$ a deterministic (possibly $z$-dependent) control parameter which satisfies
\begin{equation}
M^{-1/2} \leq \Psi \leq M^{-c}
\label{eq:def_determ_con_para}
\end{equation}
for all $z \in \mathbf D$ and for some $c>0$. 

The aim of the fluctuation averaging is to estimate linear combinations of the form $\sum_k t_{ik}X_k$
with special random variables $X_k$ and a family of complex weights $T=(t_{ik})$ that satisfy 
\begin{equation}
0 \leq \abs{t_{ik}}\leq M^{-1}, \quad \sum_k \abs{t_{ik}}\leq 1.
\label{eq:def_weight_averaging}
\end{equation}
Note that the family $T$ may be $N$-dependent.
Examples of such weights are given by $t_{ik}=s_{ik}=\E \abs{h_{ik}}^2$, $t_{ik}=N^{-1}$ or $t_{ik}=r_{ik}=\E h_{ik}^2$. 
%It will be important that $T$ commutes with $S$ in the first two cases and that $T$ commutes with $R$ in the last case.
Recall that $\Lambda(z) = \max_{x,y} \abs{G_{xy}(z) - \delta_{xy} m(z)}$ which is the basic quantity we want to estimate (cf. \eqref{eq:definition_Lambda}). 

\begin{thm}[Fluctuation Averaging] \label{thm:fluct_averag1}
Let $\mathbf D$ be a spectral domain, $\Psi$ a deterministic control parameter satisfying \eqref{eq:def_determ_con_para} and $T=(t_{ik})$ a 
weight satisfying \eqref{eq:def_weight_averaging}. If $\Lambda \prec \Psi$ then
\begin{equation}
\left|\sum_k t_{ik} \mathbb F_k \frac{1}{G_{kk}}\right| \prec \Psi^2, \quad \left|\sum_k t_{ik} \mathbb F_k G_{kk}\right| \prec \Psi^2, \quad \left|\sum_{k\neq -k} t_{ik} \mathbb F_k G_{k,-k}\right| \prec \Psi^2
\label{eq:fluct_averag1}
\end{equation}
uniformly in $i$ and $z \in \mathbf D$. If $\Lambda \prec \Psi$ and $T$ commutes with $S$ then we have
\begin{equation}
\left|\sum_k  t_{ik} (G_{kk}-m) \right| \prec \Gamma_S \Psi^2
\label{eq:fluct_averag2}
\end{equation}
uniformly in $i$ and $z \in \mathbf D$. If $\Lambda \prec \Psi$ and $T$ commutes with $R$ then we have
\begin{equation}
\left|\sum_{k\neq -k}  t_{ik} G_{k,-k} \right| \prec \Gamma_R \Psi^2
\label{eq:fluct_averag2_G_k-k}
\end{equation}
uniformly in $i$ and $z \in \mathbf D$.
\end{thm}

A similar result was proved in \cite{EJP2473}, but 
due to the fourfold symmetry we need the third estimate in \eqref{eq:fluct_averag1} and \eqref{eq:fluct_averag2_G_k-k} which were not present there.
For the first estimate in \eqref{eq:fluct_averag1}, there is the following stronger bound assuming that there is a stronger apriori bound on the off-diagonal terms,
i.e., on $\Lambda_o(z)=\max_{x\neq y}\abs{G_{xy}(z)}$ (cf. \eqref{eq:definition_Lambda}):
\begin{thm} \label{thm:fluct_averag2}
Let $\mathbf D$ be a spectral domain, $\Psi$ and $\Psi_o$ deterministic control parameters satisfying \eqref{eq:def_determ_con_para} and $T=(t_{ik})$ a 
weight satisfying \eqref{eq:def_weight_averaging}. If $\Lambda \prec \Psi$ and $\Lambda_o \prec \Psi_o$ then
\begin{equation}
\left|\sum_k t_{ik} \mathbb F_k \frac{1}{G_{kk}}\right| \prec \Psi_o^2
\label{eq:fluct_averag3}
\end{equation}
uniformly in $i$ and $z \in \mathbf D$.
\end{thm}

The proof of Theorem \ref{thm:fluct_averag1} and \ref{thm:fluct_averag2} can be found in section \ref{sec:fluct_averag}.

\subsection{Estimates on $m$ and $\Gamma$}

For convenience, we list some elementary estimates from \cite{EJP2473} which are often used in the following proofs.

\begin{lem} There is a constant $c>0$ such that for $z \in \{E+\text i \eta ; E \in [-10,10], \eta \in (0,10]\}$ we have  
\begin{equation}
	c \leq \abs{m(z)}, \quad \abs{m(z)} \leq ~ 1-c\eta,\quad \abs{m(z)} \leq ~ \eta^{-1}, \quad \Im m(z) \geq ~ c\eta. \label{eq:estimate_m}
\end{equation}
\end{lem}
Since $\Gamma\geq \Gamma_S$ it suffices to prove the following lower bounds on $\Gamma$ for $\Gamma_S$.
%The upper bound on $\Gamma$ follows immediately from $\norm{R}_{\ell^\infty \to\ell^\infty} \leq \norm{S}_{\ell^\infty \to\ell^\infty}=1$.

\begin{lem} \label{lem:estimate_Gamma} There is a constant $c>0$ such that 
\begin{equation}
	 c \leq \Gamma(z), \quad \abs{1-m^2(z)}^{-1} \leq \Gamma(z)
\label{eq:estimate_Gamma} 
\end{equation}
for all $z \in \{E+\text i \eta ; E \in [-10,10], \eta \in (0,10]\}$. 
\end{lem}

%\begin{lem} \label{Lem:Lambda_Gamma_Lipschitz}
%\marginnote{Remove this??}
%For every spectral domain $\mathbf D$, the maps $\Lambda$ and $\Gamma$ are Lipschitz-continuous on $\mathbf D$ with 
%\begin{align*}
%\abs{\Lambda(z)-\Lambda(w)}\leq  2 M^2\abs{z-w},\\
%\abs{\Gamma(z)-\Gamma(w)}\leq  2 c^{-2} M^4\abs{z-w}
%\end{align*}
%for $z, w \in \mathbf D$.
%\end{lem}

\begin{rem} \label{rem:estimate_Gamma_R}
Since $\norm{R}_{\ell^\infty\to\ell^\infty} \leq 1$ the proof of Proposition A.2 in \cite{EJP2473} yields that
\[\Gamma_R(z) \leq \frac{C\log N}{1-\max_{\pm} \left|\frac{1\pm m^2}{2}\right|} \leq \frac{C\log N}{\min\{\eta+E^2,\theta\}}\]
for $z\in\{E+\i \eta; -10 \leq E\leq 10, M^{-1}\leq \eta  \leq 10\}$ with 
\[\theta \equiv \theta(z) \defeq \begin{cases} \kappa + \frac{\eta}{\sqrt{\kappa+\eta}}, & \text{if } \abs E \leq 2,\\
	\sqrt{\kappa+\eta}, & \text{if } \abs E >2,\end{cases}\]
and $\kappa \defeq \abs{\abs{E}-2}$. 
\end{rem}

\section{Proof of the Main Result}

This section contains the proof of our main result, Theorem \ref{thm:Main_Fourfold}. %which is oriented on the line of reasoning in the chapter 5 of \cite{EJP2473}.
First, we establish the two self-consistent equations which will be the basis of all our estimates. In section 5.2, we bound the error terms in these 
self-consistent equations so that we can use them to prove a preliminary bound on the central quantity $\Lambda$ (cf. \eqref{eq:definition_Lambda}) in section 5.3. Finally, 
we complete the proof of Theorem \ref{thm:Main_Fourfold} in section 5.4 by iteratively improving the preliminary bound from the previous section.

\subsection{Self-consistent Equations}

The goal of this section is to establish the two self-consistent equations for the difference $G_{xx}-m$ and for the off-diagonal terms $G_{x,-x}$.
As the matrices are indexed by elements in $\Z\slash N\Z$ it might happen that $x = -x$ for $x \in \Z\slash N\Z$, more precisely we have $0= -0$ in 
$\Z\slash N\Z$ and moreover if $N$ is even $N/2=-N/2$. Since the expansion of the diagonal term $G_{xx}$ by means of the resolvent 
identities is a bit different for $x=-x$ and in this cases the entry $G_{x,-x}$ is in fact a diagonal term we have to distinguish the two cases, $x\neq -x$ and 
$x = -x$, in the sequel. 

Recall for the following lemma that $s_{xa}= \E \abs{h_{xa}^2}$ and $r_{xa}=\E h_{xa}^2$. 

\begin{lem} \label{lem:self_consistent_eq}
For $v_x \defeq G_{xx}-m$ we have the self-consistent equation 
\begin{equation}
-\sum_{a}s_{xa}v_a+\Upsilon_x=\frac{1}{v_x+m}-\frac 1 m
\label{eq:self_const_2}
\end{equation}
with the error term 
\[ \Upsilon_x = \begin{cases} h_{xx} + A_x - Z_x, & x = -x, \\h_{xx} +A_x +B_x -C_x -Y_x -Z_x, & x \neq -x, \end{cases}\]
 and the abbreviations
\begin{align}
A_x&\defeq \sum_{a}s_{xa}\frac{G_{ax}G_{xa}}{G_{xx}},  \hspace*{4cm} %qquad \quad 
 B_x  \defeq \sum_{a}^{(x,-x)}s_{xa}\frac{G_{a,-x}^{(x)}G_{-x,a}^{(x)}}{G_{-x,-x}^{(x)}}, 
\label{eq:definition_A_x_B_x}\\ 
C_x&\defeq \left(|h_{x,-x}|^2-s_{-x,x}\right)G_{-x,-x}^{(x)}+h_{-x,x}\sum_{a}^{(x,-x)}h_{xa}G_{a,-x}^{(x)}+h_{x,-x}\sum_{b}^{(x,-x)}G_{-x,b}^{(x)}h_{bx},
\label{eq:definition_C_x} \\
Y_x&\defeq\left(G_{-x,-x}^{(x)}\right)^{-1}\sum_{a,b}^{(x,-x)}h_{xa}G_{a,-x}^{(x)}G_{-x,b}^{(x)}h_{bx},\label{eq:definition_Y_x} 
\qquad  %\\
Z_x\defeq \begin{cases} %\displaystyle 
\sum_{a,b}^{(x)} \F_x \left[h_{xa}G_{ab}^{(x)}h_{bx} \right], & x = -x, \\ 
%\displaystyle 
\sum_{a,b}^{(x,-x)}\mathbb F_x\left[h_{xa}G_{ab}^{(x,-x)}h_{bx} \right] , & x\neq -x .
\end{cases} %\nonumber
\end{align}
The self-consistent equation for $G_{x,-x}$ is given by 
\begin{equation}
G_{x,-x}=m^2\sum_{a\neq -a}r_{xa} G_{a,-a}+\e_x,
\label{eq:self_averaging_G_x-x}
\end{equation}
for $x \neq -x$ where we defined $\e_x\defeq \e^1_x +\e^2_x-\e^3_x-\e^4_x$ with the error terms 
\begin{align*}
\e^1_x \defeq & -m^2\sum_{a\in \{x,-x\}} r_{xa}G_{a,-a} +  m^2\sum_{a=-a} r_{xa}G_{aa}+ \left(G_{xx}G_{-x,-x}^{(x)}-m^2\right)\sum_a^{(x,-x)}r_{xa}G_{a,-a} - G_{xx}G_{-x,-x}^{(x)}h_{x,-x} , \\
\e^2_x \defeq  & G_{xx}G_{-x,-x}^{(x)}\sum_a^{(x,-x)} \F_x \left[h_{xa}G_{ab}^{(x,-x)} h_{b,-x}\right], %\quad  
%= G_{xx}G_{-x,-x}^{(x)}\sum_a^{(x,-x)}\left(h_{xa}^2-r_{xa}\right)G_{a,-a}^{(x,-x)}+ G_{xx}G_{-x,-x}^{(x)}\sum_{a\neq b}^{(x,-x)}h_{xa}G_{a,-b}^{(x,-x)}h_{xb},
\\ \e^3_x \defeq &  G_{-x,-x}^{(x)}\sum_a^{(x,-x)}r_{xa}G_{ax}G_{x,-a}, \qquad \qquad \e^4_x \defeq G_{xx}\sum_a^{(x,-x)}r_{xa}G_{a,-x}^{(x)}G_{-x,-a}^{(x)}.
\end{align*}
\end{lem}

%TODO change the estimate for G_{x,-x} in the following

%TODO explanations

The self-consistent equation \eqref{eq:self_const_2} has the same form as (5.9) in \cite{EJP2473} and it is proved in a similar way by expanding by means of Schur's complement 
formula and computing the partial expectation of a term in this expansion. 
However, we had to replace $P_i$ by $\E_x$ to derive it and the error term $\Upsilon_x$ contains terms which did not appear in (5.8) from \cite{EJP2473}. 
(If $x=-x$ then $\Upsilon_x$ has the same form as in \cite{EJP2473}.) 
The term $A_x$ is exactly the same as $A_i$ in (5.8) of 
\cite{EJP2473}. The term $Z_x$ is the analogue of $Z_i$ in \cite{EJP2473} but the terms $B_x$, $C_x$ and $Y_x$ are completely new
and will require new estimates.

%TODO complete explanations for the second self-consistent equation
The self-consistent equation \eqref{eq:self_averaging_G_x-x} is new and does not have a counterpart in \cite{EJP2473}. Due to the fourfold symmetry there is the necessity to 
introduce it since in contrast to the symmetry studied in \cite{EJP2473} proving directly that the off-diagonal elements $G_{x,-x}$ are small is not possible. 

As deriving this self-consistent equation follows the same line as the proof of \eqref{eq:self_const_2} -- expanding and computing the partial expectation of a term in this expansion -- 
it is not surprising that some error terms in \eqref{eq:self_averaging_G_x-x} have counterparts in \eqref{eq:self_const_2}. Namely, $\e^2_x$ is the counterpart of $Z_x$. Moreover, 
 $\e^3_x$ and $\e^4_x$ are the error terms corresponding to $A_x$ and $B_x$, respectively.
%Note that $\e^2_x = G_{xx}G_{-x,-x}^{(x)}\sum_a^{(x,-x)} \F_x \left[h_{xa}G_{ab}^{(x,-x)} h_{b,-x}\right]$

\begin{proof}
We start with the proof of \eqref{eq:self_const_2}. 
For $x=-x$ the derivation of \eqref{eq:self_const_2} follows exactly as (5.9) in section 5.1 of \cite{EJP2473} since $\E_x$ and $\F_x$ agree with $P_x$ and $Q_x$ respectively in this case. 
Similarly, for $x \neq -x$ the self-consistent equation \eqref{eq:self_const_2} will be obtained from Schur's complement formula \eqref{eq:Schur_formula} with $\mathbb T=\emptyset$. %TODO replace MANIPULATING
In this case, its last term can be written in the form
\begin{align}
\sum_{a,b}^{(x)}h_{xa}G_{ab}^{(x)}h_{bx}= &
h_{x,-x}G_{-x,-x}^{(x)}h_{-x,x}+\sum_{a}^{(x,-x)}h_{xa}G_{a,-x}^{(x)}h_{-x,x}+\sum_{b}^{(x,-x)}h_{x,-x}G_{-x,b}^{(x)}h_{bx}\nonumber\\
&+\sum_{a,b}^{(x,-x)}
h_{xa}G_{ab}^{(x,-x)}h_{bx} + \left(G_{-x,-x}^{(x)}\right)^{-1}\sum_{a,b}^{(x,-x)}h_{xa}G_{a,-x}^{(x)}G_{-x,b}^{(x)}h_{bx} 
\label{eq:expansion_expectation_Schur}
\end{align}
by applying the resolvent identity \eqref{eq:resolvent_identity1}. 
Since the random variables $h_{xa}$ and $h_{-x,b}$ are independent of $H^{(x,-x)}$ we have
$\E_x\left[h_{xa}G_{ab}^{(x,-x)}h_{bx} \right]=s_{xa}G_{aa}^{(x,-x)}\delta_{ab}.$ 
\noindent Thus, 
\begin{align}
\sum_{a,b}^{(x,-x)}\E_x\left[h_{xa}G_{ab}^{(x,-x)}h_{bx} \right]=& \sum_{a}^{(x,-x)}s_{xa}G_{aa}^{(x,-x)}\nonumber\\ =&
\sum_{a}s_{xa}G_{aa}-\sum_{a}s_{xa}\frac{G_{ax}G_{xa}}{G_{xx}} -s_{-x,x}G_{-x,-x}^{(x)} - \sum_{a}^{(x,-x)}s_{xa}\frac{G_{a,-x}^{(x)}G_{-x,a}^{(x)}}{G_{-x,-x}^{(x)}}, \nonumber
\end{align}
where we used in the second step the resolvent identity \eqref{eq:resolvent_identity1} twice. 
By splitting the fourth summand on the right-hand side of \eqref{eq:expansion_expectation_Schur} according to $\E_x +\F_x=1$, we get 
\begin{align}
\sum_{a,b}^{(x,-x)}h_{xa}G_{ab}^{(x,-x)}h_{bx} = &\sum_{a,b}^{(x,-x)}\E_x\left[h_{xa}G_{ab}^{(x,-x)}h_{bx} \right]+\sum_{a,b}^{(x,-x)}\mathbb F_x\left[h_{xa}G_{ab}^{(x,-x)}h_{bx} \right]\nonumber \\
=&\sum_{a}s_{xa}G_{aa}-A_x-s_{-x,x}G_{-x,-x}^{(x)}-B_x+Z_x.
\label{eq:expansion_term4}
\end{align}
Therefore, the results of \eqref{eq:expansion_expectation_Schur} and \eqref{eq:expansion_term4} allow us to write \eqref{eq:Schur_formula} in the form
\begin{equation*}
 \frac{1}{G_{xx}}=-z-m+\Upsilon_x -\sum_a s_{xa}v_a,
%\label{eq:self_const_1}
\end{equation*}
which implies \eqref{eq:self_const_2} using \eqref{eq:semicircle2}.

We fix $x \neq -x$. To derive \eqref{eq:self_averaging_G_x-x} we apply the resolvent identity 
\eqref{eq:resolvent_identity2} twice to get 
\begin{equation}
G_{x,-x} = - G_{xx}G_{-x,-x}^{(x)} h_{x,-x} + G_{xx} G_{-x,-x}^{(x)} \sum_{a,b}^{(x,-x)} h_{xa}G_{ab}^{(x,-x)} h_{b,-x}.
\label{eq:G_offdiag_ansatz}
\end{equation}
Since $\E_x h_{xa}G_{ab}^{(x,-x)} h_{b,-x} = G_{a,-a}^{(x,-x)} r_{xa} \delta_{b,-a}$ splitting up the sum in the second term in \eqref{eq:G_offdiag_ansatz} according to $\E_x + \F_x=1$ yields 
\begin{equation}
G_{x,-x} = - G_{xx}G_{-x,-x}^{(x)} h_{x,-x} + G_{xx} G_{-x,-x}^{(x)} \sum_{a}^{(x,-x)} r_{xa} G_{a,-a} +\e^2_x-\e^3_x -\e^4_x
\label{eq:G_offdiag_derivation2}
\end{equation}
where we used the resolvent identity \eqref{eq:resolvent_identity1} twice. We obtain \eqref{eq:self_averaging_G_x-x} by adding and substracting $m^2\sum_a r_{xa} G_{a,-a}$ to the right-hand side of 
\eqref{eq:G_offdiag_derivation2}.
\end{proof}

%\begin{equation}
%Z_x\defeq \begin{cases} \sum_{a,b}^{(x)} \F_x [h_{xa}G_{ab}^{(x)}h_{bx} ] = \sum_{a}^{(x)} \left(|h_{xa}|^2-s_{xa}\right)G_{aa}^{(x)} + \sum_{a\neq b}^{(x)}h_{xa}G_{ab}^{(x)}h_{bx}, & x = -x, \\ 
%\end{cases}
%\end{equation}

\subsection{Auxiliary Estimates}

The next lemma contains bounds on the resolvent entries of minors of $H$ if there exists an apriori bound on $\Lambda$ (Recall its definition 
in \eqref{eq:definition_Lambda}). We will use a deterministic (possibly $z$-dependent) parameter $\Psi$ which fulfills 
\begin{equation}
cM^{-\frac 1 2} \leq \Psi \leq M^{-c}
\label{eq:deterministic_control_general}
\end{equation}
for some $c>0$ and all large enough $N$.

\begin{lem}
\label{lem:aux_estimate_G_ij}
Let $\mathbf D$ be a spectral domain and $\varphi$ the indicator function of a (possibly $z$-dependent) event. Let $\Psi$ be a deterministic control 
parameter satisfying \eqref{eq:deterministic_control_general}. If $\varphi \Lambda \prec \Psi$ and $\mathbb T \subset \N$ is a fixed finite subset then
$$%\begin{equation}
\varphi |G_{ij}^{(\mathbb T)}| \prec \varphi \Lambda_o \prec \Psi, \quad \varphi |G_{ii}^{(\mathbb T)}| \prec 1,\quad \frac{\varphi}{|G_{ii}^{(\mathbb T)}|} \prec 1, \quad 
\varphi\abs{G_{ii}^{(\mathbb T)}-m} \prec \varphi \Lambda, \quad \varphi \Im G_{ii}^{(\mathbb T)} \prec \Im m+\Lambda
%\label{eq:bound_Im_G_aa^T}
%\label{eq:estimate_diagonal_term_T_Lambda}
$$%\end{equation}
uniformly in $z \in \mathbf D$ and in $i,j$ for $i\neq j$ and $i,j \notin \mathbb T$.
\end{lem}

\begin{proof}
This result follows by induction on the size of $\mathbb T$ using \eqref{eq:estimate_m} and \eqref{eq:resolvent_identity1}. 
\end{proof}

Using this result we will establish the first bounds on the error terms in the self-consistent equations in the next lemma.
When applying the first part of the following lemma the indicator $\varphi$ will be defined precisely in such way that the condition $\varphi \Lambda \prec M^{-c}$ holds, i.e., 
to ensure that $\varphi\Lambda$ is small. 

\begin{lem} \label{lem:aux_estimates_fourfold}Let $\mathbf D$ be a spectral domain.
\begin{enumerate}[(i)]
\item If $\varphi$ is an indicator function %of a (possibly $z$-dependent) event 
such that $\varphi\Lambda \prec M^{-c}$ (for some $c>0$) then
\begin{align}
\varphi(\Lambda_g+\abs{A_x}+\abs{B_x}+\abs{C_x}+\abs{Y_x}+\abs{Z_x})&\prec \varphi \Lambda^2+\sqrt{\frac{\Im m +\Lambda}{M\eta}},\label{eq:aux_est_lambda_o_Ups_x_phi} \\
\varphi(\abs{\e_x^1}+\abs{\e_x^2}+\abs{\e_x^3}+\abs{\e_x^4}) &\prec \varphi \Lambda^2 +\sqrt{\frac{\Im m +\Lambda}{M\eta}}\label{eq:aux_est_mathcal_E_phi}
\end{align}
uniformly in $x$ and $z\in \mathbf D$.
\item For fixed $\eta>0$ we have the estimates
\begin{equation}
\Lambda_- \leq  \eta^{-2}\Lambda_-+2\eta^{-3}\Lambda_-^2+\epsilon 
\label{eq:inequality_Lambda-}
\end{equation}
with $\epsilon\prec M^{-1/2}$ uniformly in $z \in \{w \in \C; \Im w = \eta\}$, and 
\begin{align}
\Lambda_g \prec & ~M^{-1/2}+\Lambda_-, \label{eq:estimate_Lambda_g_eta_fixed}\\
\abs{A_x}+\abs{B_x}+\abs{C_x}+\abs{Y_x}+\abs{Z_x} \prec & ~M^{-1/2}+\Lambda_o \label{eq:Upsilon_for_constant_eta}
\end{align}
uniformly in $x$ and in $z\in \{w \in \C; \Im w = \eta\}$.
%\norm{\Upsilon}_\infty&\prec M^{-1/2}+\Lambda_-^2+\Lambda_o,\label{eq:Upsilon_for_constant_eta}\\
%\Lambda_- &\leq \eta^{-2}\Lambda_-+\eta^{-1}\Lambda_o^2+\norm{(\e_x^3)_x}_{\infty}+\norm{(\e_x^4)_x}_{\infty}+\norm{(\e_x^5)_x}_{\infty}+\norm{(\e_x^6)_x}_{\infty},\label{eq:lambda_-_constant_eta} \\
%\Lambda_- &\leq \eta^{-2}\Lambda_-+\eta^{-1}\Lambda_o^2+\norm{\e^3}_{\infty}+\norm{\e^4}_{\infty}+\norm{\e^5}_{\infty}+\norm{\e^6}_{\infty},\label{eq:lambda_-_constant_eta} \\
%\Lambda_o &\leq \tilde \e+ \eta^{-1}\Lambda_-, \label{eq:lambda_o-constant_eta_det}\\
%\Lambda_o & \prec M^{-1/2} + \eta^{-1}\Lambda_-\label{eq:lambda_o-constant_eta_prec}
%\end{align}
%uniformly in $z \in \{w \in \C; \Im w = \eta\}$ with $\tilde \e \prec M^{-1/4}$ uniformly $z\in \{w \in \C; \Im w = \eta\}$ and \[\abs{\e^3_x}+\abs{\e^4_x}+\abs{\e^5_x}+\abs{\e^6_x}\prec M^{-1/2}\] 
\end{enumerate}
\end{lem}

\begin{proof}
In this proof we will occasionally split the index set of a summation into the parts $\{a\neq -a\}$ and $\{a = -a \}$
and use that the latter set contains at most two elements.

In the following proof of the first part Lemma \ref{lem:aux_estimate_G_ij} will be applied several times with $\Psi=M^{-c}$. Note that $M^{-1/2}\prec \sqrt{(\Im m +\Lambda)/(M\eta)}$ because of the fourth estimate 
in \eqref{eq:estimate_m}.
First, we assume $x\neq -x$. 
Applying the second estimate in \eqref{eq:estimate_h_xy_stoch_dom} and \eqref{eq:condition_s_xy} to the definition of $A_x$ in \eqref{eq:definition_A_x_B_x} yields 
\begin{equation}
\varphi \abs{A_x} \prec s_{xx}\abs{G_{xx}}+ \sum_a^{(x)} s_{xa}\varphi \frac{\abs{G_{xa}G_{ax}}}{\abs{G_{xx}}} \prec M^{-1}+ \varphi \Lambda_o^2.
\label{eq:estimate_A_x}
\end{equation}
Similarly, using the first estimate in Lemma \ref{lem:aux_estimate_G_ij} we get $\varphi\abs{B_x}\prec \varphi\Lambda_o^2$.

The representation 
\begin{equation}
C_x=|h_{x,-x}|^2G_{-x,-x}^{(x)}-s_{-x,x}G_{-x,-x}^{(x)}-\frac{G_{x,-x}}{G_{xx}}h_{-x,x}-h_{x,-x}\frac{G_{-x,x}}{G_{xx}},
\label{eq:second_representation_C_x}
\end{equation}
which follows from the resolvent identity \eqref{eq:resolvent_identity2}, together with \eqref{eq:estimate_h_xy_stoch_dom} implies 
\begin{equation}
\varphi \abs{C_x} \prec M^{-1/2}.
\label{eq:est_C_x}
\end{equation}

To estimate $Y_x$ we need the following two auxiliary bounds: We have
\begin{equation}
\varphi\left|\sum_a^{(x,-x)} h_{xa}^2 G_{a,-a}^{(x,-x)}\right| \leq \sum_{a\neq -a}^{(x,-x)} \abs{h_{xa}}^2 \varphi \abs{G_{a,-a}^{(x,-x)}} + \sum_{a=-a}^{(x,-x)} \abs{h_{xa}}^2 \varphi \abs{G_{aa}^{(x,-x)}} 
\prec \varphi \Lambda_o +M^{-1},
\label{eq:estimate_h_xk^2G_k-k^x-x}
\end{equation}
where we used \eqref{eq:estimate_h_xy_stoch_dom} and \eqref{eq:condition_s_xy} in last step. 
Now, we use the quadratic Large Deviation Bounds from \cite{EJP2473} after conditioning on $G^{(x,-x)}$.
By applying (C.4) in \cite{EJP2473} with $X_k=\zeta_{xk}$ and $a_{kl}=s_{xk}^{1/2}G_{k,-l}^{(x,-x)}s_{xl}^{1/2}$ we get
\begin{equation}
\varphi\left|\sum_{k \neq l}^{(x,-x)}h_{xk}G_{k,-l}^{(x,-x)}h_{xl}\right|^2
\prec \sum_{k \neq l}^{(x,-x)}s_{xk} s_{xl} \varphi\abs{G_{k,-l}^{(x,-x)}}^2 \prec  \frac{\varphi}{M\eta}\sum_{k}^{(x,-x)}s_{xk} \Im G_{kk}^{(x,-x)}\prec \frac{\Im m + \Lambda}{M\eta},
\label{eq:bound_h_xk_G_kl^x-x_h_l-x_k_neq_l}
\end{equation}
where we used the second estimate in \eqref{eq:estimate_h_xy_stoch_dom} and \eqref{eq:Ward_identity} in the second step. %TODO maybe insert the usage of the last estimate in Lemma 5.1
Thus, the representation 
\begin{equation}
Y_x=G_{-x,-x}^{(x)}\left(\sum_{a,k}^{(x,-x)} h_{xa}G_{ak}^{(x,-x)}h_{k,-x}\right)\left(\sum_{b,l}^{(x,-x)}h_{-x,l}G_{lb}^{(x,-x)}h_{bx}\right),
\label{eq:expansion_Y_x}
\end{equation}
which follows from the resolvent identity \eqref{eq:resolvent_identity2}, yields (after separating the case $k=-a$)
\begin{equation}
\varphi\abs{Y_x} \prec \varphi\left|\sum_a^{(x,-x)} h_{xa}^2 G_{a,-a}^{(x,-x)}\right|^2+\varphi\left|\sum_{a \neq k}^{(x,-x)}h_{xa}G_{a,-k}^{(x,-x)}h_{xk}\right|^2 \prec \varphi\Lambda_o^2 + \frac{\Im m + \Lambda}{M\eta} 
\prec \varphi\Lambda_o^2 + \sqrt\frac{\Im m + \Lambda}{M\eta}.
\label{eq:est_Y_x_varphi}
\end{equation}
%
%
%Now, we prepare the estimate of $Z_x$ by first noting that  
Before estimating $Z_x$ we note that 
\[Z_x \defeq \begin{cases} \displaystyle \sum_{a,b}^{(x,-x)}\mathbb F_x\left[h_{xa}G_{ab}^{(x,-x)}h_{bx} \right] =\sum_{a}^{(x,-x)}\left(|h_{xa}|^2-s_{xa}\right)G_{aa}^{(x,-x)} 
+\sum_{a\neq b}^{(x,-x)}h_{xa}G_{ab}^{(x,-x)}h_{bx}, & x \neq -x, \\
\displaystyle \sum_{a,b}^{(x)} \F_x [h_{xa}G_{ab}^{(x)}h_{bx} ] = \sum_{a}^{(x)} \left(|h_{xa}|^2-s_{xa}\right)G_{aa}^{(x)} + \sum_{a\neq b}^{(x)}h_{xa}G_{ab}^{(x)}h_{bx}, & x = -x. \end{cases} \] 
We fix $x\neq -x$ and apply (C.4) in \cite{EJP2473} with $X_i=\zeta_{xi}$ and $a_{ij}=s_{xi}^{1/2} G_{ij}^{(x,-x)}s_{jx}^{1/2}$ to get
\begin{equation}
\varphi \left|\sum_{i\neq j}^{(x,-x)} h_{xi}G_{ij}^{(x,-x)}h_{jx}\right|^2 \prec \left(\sum_{ i \neq j}^{(x,-x)} s_{xi}s_{jx} \varphi\abs{G_{ij}^{(x,-x)}}^2\right)^{1/2} \prec \frac{\Im m + \Lambda} {M\eta},
\label{eq:first_estimate_X_x}
\end{equation}
where the last step follows in the same way as the last step in \eqref{eq:bound_h_xk_G_kl^x-x_h_l-x_k_neq_l}.  
Moreover, (C.2) in \cite{EJP2473} with $X_i=(\abs{\zeta_{xi}}^2-1)(\E\abs{\zeta_{xi}}^4-1)^{-1/2}$ and $a_i=(\E\abs{\zeta_{xi}}^4-1)^{1/2}s_{xi}G_{ii}^{(x,-x)}$
implies
\begin{equation}
\varphi \left|\sum_{i}^{(x,-x)}\left(|h_{xi}|^2-s_{xi}\right)G_{ii}^{(x,-x)}\right|^2 \prec \sum_i^{(x,-x)}s_{xi}^2 (\E|\zeta_{xi}|^4-1)
\varphi\abs{G_{ii}^{(x,-x)}}^2 %\nonumber \\ \prec & (\mu_4-1)^{1/2} M^{-1/2}\left(\sum_i s_{xi}\right)^{1/2} 
\prec M^{-1},
\label{eq:est_sum_abs_h^2-s_G_ii}
\end{equation}
where we used \eqref{eq:finite_moments}, the second estimate in \eqref{eq:estimate_h_xy_stoch_dom} and \eqref{eq:condition_s_xy} in the last step.
Therefore, absorbing $M^{-1/2}$ into the second summand we get 
\begin{equation}
\varphi \abs{Z_x}\leq \varphi \left|\sum_{i\neq j}^{(x,-x)} h_{xi}G_{ij}^{(x,-x)}h_{jx}\right| +\varphi \left|\sum_{i}^{(x,-x)}\left(|h_{xi}|^2-s_{xi}\right)
G_{ii}^{(x,-x)}\right| \prec \sqrt\frac{\Im m + \Lambda} {M\eta}.
\label{eq:est_Z_x_varphi}
\end{equation}
If $x=-x$ then $Z_x$ can be bounded by the right-hand side in \eqref{eq:aux_est_lambda_o_Ups_x_phi} similarly to the previous estimate and for $A_x$ in exactly the same way as in \eqref{eq:estimate_A_x}.

%generic off-diagonal term
To estimate the generic off-diagonal entry $G_{xy}$ under the assumption that all of $x,-x,y,-y$ are different, we use the expansion
\begin{align}
G_{xy}=& -G_{xx}^{(-x,-y)}G_{yy}^{(x,-x,-y)}\left(h_{xy}
-\sum_{k,l}^{(x,-x,y,-y)}h_{xk}G_{kl}^{(x,-x,y,-y)}h_{ly}\right)+\frac{G_{x,-y}^{(-x)}G_{-y,y}^{(-x)}}{G_{-y,-y}^{(-x)}}+\frac{G_{x,-x}G_{-x,y}}{G_{-x,-x}},
\label{eq:exp_G_xy^(-x,-y)}
\end{align}
which follows from applying \eqref{eq:resolvent_identity2} twice and afterwards applying the first identity in \eqref{eq:resolvent_identity1} twice.
Conditioning on $G^{(x,-x,y,-y)}$ and applying (C.3) in \cite{EJP2473} with $X_k=\zeta_{xk}$, $Y_l=\zeta_{ly}$ and 
$a_{kl}=s_{xk}^{1/2} G_{kl}^{(x,-x,y,-y)} s_{ly}^{1/2}$ yield 
\begin{equation}
\varphi\left|\sum_{k,l}^{(x,-x,y,-y)}h_{xk}G_{kl}^{(x,-x,y,-y)}h_{ly}\right|^2\prec\varphi\sum_{k,l}^{(x,-x,y,-y)}s_{xk}\abs{G_{kl}^{(x,-x,y,-y)}}^2 s_{ly}
\prec {\frac{\Im m +\Lambda} {M\eta}},
\label{eq:est_first_term_exp_G_xy}
\end{equation}
where the last step follows exactly as in \eqref{eq:bound_h_xk_G_kl^x-x_h_l-x_k_neq_l}, which implies
\[
\varphi\abs{G_{xy}} \prec M^{-1/2}+\sqrt\frac{\Im m +\Lambda} {M\eta} + \varphi \Lambda_o^2. \]
If $x=-x$ or $y=-y$ then the proof of the last statement is easier. 
 This finishes the proof of \eqref{eq:aux_est_lambda_o_Ups_x_phi}.

Now, we turn to the proof of \eqref{eq:aux_est_mathcal_E_phi}. The trivial estimate $\abs{\E h_{xy}^2} \leq \E \abs{h_{xy}}^2=s_{xy} \leq M^{-1}$ implies that the first two terms in $\varphi\abs{\e_x^1}$ 
are bounded by $M^{-1}$. By \eqref{eq:estimate_h_xy_stoch_dom} its last term is bounded by $M^{-1/2}$. 
Splitting the summation in the third term of $\varphi\abs{\e_x^1}$ into $a\neq -a$ and $a=-a$ and using the estimate on $\abs{\E h_{xy}}^2$ 
we obtain $\varphi\abs{\e^1_x}\prec \varphi \Lambda \Lambda_-+M^{-1/2}$ due to \eqref{eq:condition_s_xy},
\eqref{eq:estimate_m}, the fourth estimate in Lemma \ref{lem:aux_estimate_G_ij} and \eqref{eq:estimate_h_xy_stoch_dom}. 
%\else
%and $\varphi\abs{\e^2_x}\prec \varphi \Lambda \Lambda_-+M^{-1}$ due to \eqref{eq:condition_s_xy}, \eqref{eq:estimate_m} and the fourth estimate in Lemma \ref{lem:aux_estimate_G_ij}.
%\fi
Similarly to the bound on the third term in $\varphi\abs{\e_x^1}$, we get $\varphi\abs{\e_x^3}\prec \varphi \Lambda_o^2$ and $\varphi\abs{\e_x^4}\prec \varphi \Lambda_o^2$.
To estimate $\e_x^2$ we calculate the partial expectation in its definition which yields  
\[
\e_x^2=%G_{xx}G_{-x,-x}^{(x)}\sum_a^{(x,-x)} \F_x \left[h_{xa}G_{ab}^{(x,-x)} h_{b,-x}\right]= 
 G_{xx}G_{-x,-x}^{(x)}\sum_a^{(x,-x)}\left(h_{xa}^2-r_{xa}\right)G_{a,-a}^{(x,-x)}+ G_{xx}G_{-x,-x}^{(x)}\sum_{a\neq b}^{(x,-x)}h_{xa}G_{a,-b}^{(x,-x)}h_{xb}.
\]
Similarly to \eqref{eq:est_sum_abs_h^2-s_G_ii} the first term can be bounded by $M^{-1}$. 
%The estimate 
%\[ \varphi\left|\sum_a(h_{xa}^2-\E h_{xa}^2)G_{a,-a}^{(x,-x)}\right|^2 \prec M^{-1} \]
%can be proved similarly to \eqref{eq:est_sum_abs_h^2-s_G_ii}. 
%Together with \eqref{eq:estimate_h_xy_stoch_dom} it implies the bound $\varphi\abs{\e_x^3}\prec M^{-1/2}$. 
Using \eqref{eq:bound_h_xk_G_kl^x-x_h_l-x_k_neq_l} for the second term implies \[\varphi\abs{\e_x^2} \prec \sqrt\frac{\Im m+\Lambda}{M\eta}\] which 
completes the proof of \eqref{eq:aux_est_mathcal_E_phi}.

Finally, we prove part (ii) of Lemma \ref{lem:aux_estimates_fourfold}. In contrast to part (i), we fix $\eta>0$. 
%Thus, we can treat $\eta$ as a constant which is irrelevant in the estimates with respect to the stochastic domination.
Since constants do not matter in the estimates with respect to the stochastic domination we will not keep track of $\eta$ in such estimates.
We start the proof of part (ii) of Lemma \ref{lem:aux_estimates_fourfold} with verifying \eqref{eq:Upsilon_for_constant_eta}. 
First, we remark that applying \eqref{eq:estimate_h_xy_stoch_dom}, \eqref{eq:Ward_identity} and \eqref{eq:deterministic_bound_G_ij^T} yields 
\begin{equation}
\left|\sum_{a}^{(\mathbb T)}h_{xa}G_{ab}^{(\mathbb T')}\right|\leq \left(\sum_a\abs{h_{xa}}^2\right)^{1/2}\left(\sum_{a}\abs{G_{ab}^{(\mathbb T')}}^2\right)^{1/2}
\prec \left(\sum_a s_{xa} \right)^{1/2}\left(\eta^{-1}\Im G_{bb}^{(\mathbb T')}\right)^{1/2} \leq \eta^{-1}
\label{eq:estimate_sum_a_h_xa_G_ab^(x)}
\end{equation}
for arbitrary finite subsets $\mathbb T, \mathbb T' \subset \mathbb N$.
The resolvent identity \eqref{eq:resolvent_identity2} and the previous bound imply 
\begin{equation} 
\abs{A_x}\leq \abs{s_{xx}G_{xx}} + \sum_a^{(x)}s_{xa}\abs{G_{ax}}\left|\sum_b^{(x)}h_{xb}G_{ba}^{(x)}\right|
\prec M^{-1} +\Lambda_o,
\label{eq:aux_est_A_x_2}
\end{equation}
where we used \eqref{eq:estimate_h_xy_stoch_dom} and \eqref{eq:deterministic_bound_G_ij^T} in the second step.
The estimate 
\begin{equation}
\abs{B_x} \leq \sum_a^{(x,-x)} s_{xa} \left|\sum_k^{(x,-x)}G_{ak}^{(x,-x)}h_{k,-x}\right|  \abs{G^{(x)}_{-x,a}} \prec M^{-1/2}
\label{eq:aux_est_B_x_2}
\end{equation}
is a consequence of (C.2) in \cite{EJP2473}  with $X_k=\zeta_{k,-x}$ and $a_k=s_{k,-x}^{1/2}G_{ak}^{(x,-x)}$, \eqref{eq:Ward_identity}, \eqref{eq:deterministic_bound_G_ij^T} 
and \eqref{eq:condition_s_xy}.

Applying \eqref{eq:estimate_sum_a_h_xa_G_ab^(x)} to the second and third term in \eqref{eq:definition_C_x} and \eqref{eq:estimate_h_xy_stoch_dom} to the first term yields $\abs{C_x} \prec M^{-1/2}$.

To estimate $Y_x$ we start from \eqref{eq:expansion_Y_x} but \eqref{eq:estimate_h_xk^2G_k-k^x-x} is estimated differently.
Using the resolvent identity \eqref{eq:resolvent_identity1} twice we get
\begin{align*}
\left|\sum_k^{(x,-x)} h_{xk}^2 G_{k,-k}^{(x,-x)}\right| \prec  & \quad \sum_{k\neq -k}^{(x,-x)} s_{xk} \abs{G_{k,-k}} +\sum_{k=-k}^{(x,-x)} s_{xk} \abs{G_{kk}} \\
 & + \sum_k^{(x,-x)} s_{xk} \frac{\abs{G_{k,-x}^{(x)}G_{-x,-k}^{(x)}}}{\abs{G_{-x,-x}^{(x)}}} + \sum_k^{(x,-x)} s_{xk} \frac{\abs{G_{kx}G_{x,-k}}}{\abs{G_{xx}}} \prec \Lambda_o + M^{-1/2},
\end{align*}
where the last step follows similarly to \eqref{eq:aux_est_A_x_2} and \eqref{eq:aux_est_B_x_2}.
Combining this with the usage of \eqref{eq:deterministic_bound_G_ij^T} instead of Lemma \ref{lem:aux_estimate_G_ij} in \eqref{eq:bound_h_xk_G_kl^x-x_h_l-x_k_neq_l} yields $\abs{Y_x}\prec M^{-1/2} + \Lambda_o$.
We get $\abs{Z_x} \prec M^{-1/2}$ by similar adjustments of \eqref{eq:est_Z_x_varphi}.
This completes the proof of \eqref{eq:Upsilon_for_constant_eta}.

Before proving \eqref{eq:inequality_Lambda-} we show 
\begin{equation} 
\Lambda_g \leq \eta^{-1} \Lambda_-+ \tilde\epsilon 
\label{eq:lem_fixed_eta_aux_claim}
\end{equation} 
with some $\tilde\epsilon \prec M^{-1/4}$ uniformly for $z \in \{w\in \C; \Im w=\eta\}$. In case all of $x$, $-x$, $y$ and $-y$ are different it 
will be derived from the representation in \eqref{eq:exp_G_xy^(-x,-y)}. For the fourth term in \eqref{eq:exp_G_xy^(-x,-y)} we obtain 
\begin{align}
\frac{\abs{G_{x,-x}G_{-x,y}}}{\abs{G_{-x,-x}}}\leq &\Lambda_-\left|\sum_a^{(-x)}h_{-x,a}G_{ay}^{(-x)}\right| \leq \Lambda_-\left(\sum_a^{(-x)}\abs{h_{-x,a}}^2\right)^{1/2} \left(\sum_a^{(-x)} \abs{G_{ay}^{(-x)}}^2\right)^{1/2}
\nonumber \\
\leq & \Lambda_-\eta^{-1}+\eta^{-2}\left|\sum_a^{(-x)}(\abs{h_{-x,a}}^2-s_{-x,a})\right|^{1/2} \label{eq:estimate_third_term_G_xy}
\end{align}
by applying the resolvent identity \eqref{eq:resolvent_identity2} and inserting $s_{-x,a}$.
In the last step, we applied \eqref{eq:Ward_identity} and \eqref{eq:deterministic_bound_G_ij^T}. Note that similarly to \eqref{eq:est_sum_abs_h^2-s_G_ii} we conclude that the second term is dominated by $M^{-1/4}$.
For the third summand in \eqref{eq:exp_G_xy^(-x,-y)} we use the estimate
\[
\left| \frac{G_{x,-y}^{(-x)}G_{-y,y}^{(-x)}}{G_{-y,-y}^{(-x)}}\right| = \left|G_{xx}^{(-x)}\sum_a^{(x,-x)} h_{xa}G_{a,-y}^{(x,-x)}\right| \left| \sum_a^{(-y,-x)} h_{-y,a}G_{ay}^{(-y,-x)}\right| \prec M^{-1/2},
\]
where we used (C.2) in \cite{EJP2473} as in the proof of \eqref{eq:aux_est_B_x_2} for the first factor and \eqref{eq:estimate_sum_a_h_xa_G_ab^(x)} for the second factor.
The first summand in \eqref{eq:exp_G_xy^(-x,-y)} is bounded by $M^{-1/2}$ due to \eqref{eq:estimate_h_xy_stoch_dom} and \eqref{eq:deterministic_bound_G_ij^T}.
Using \eqref{eq:deterministic_bound_G_ij^T} instead of Lemma \ref{lem:aux_estimate_G_ij} in \eqref{eq:est_first_term_exp_G_xy} yields that the second term in \eqref{eq:exp_G_xy^(-x,-y)} is dominated by $M^{-1/2}$ as well. 

We denote the sum of the absolute values of the first three summands in \eqref{eq:exp_G_xy^(-x,-y)} and the second summand in \eqref{eq:estimate_third_term_G_xy} by $\tilde\epsilon_{xy}$ and set $\tilde\epsilon \defeq 
\sup_{x,y} \tilde\epsilon_{xy}$.  Then the above considerations show $\tilde\epsilon \prec M^{-1/4}$ 
in this case. If $x=-x$ or $y=-y$ then estimating $G_{xy}$ is easier. Thus, \eqref{eq:lem_fixed_eta_aux_claim} follows.

Without inserting $s_{-x,a}$ in \eqref{eq:estimate_third_term_G_xy} and instead using \eqref{eq:estimate_h_xy_stoch_dom} we see that the representation \eqref{eq:exp_G_xy^(-x,-y)} implies 
\eqref{eq:estimate_Lambda_g_eta_fixed}.

To prove \eqref{eq:inequality_Lambda-} we assume $x\neq -x$ and consider the expansion
\[G_{x,-x} = G_{xx}G_{-x,-x}^{(x)}\sum_{a\neq -a}^{(x,-x)} r_{xa} G_{a,-a} + G_{xx}G_{-x,-x}^{(x)}\sum_{a=-a}^{(x,-x)} r_{xa} G_{a,-a}- G_{xx}G_{-x,-x}^{(x)}h_{x,-x} + \e_x^2 - \e_x^3 -\e_x^4 . \]
Obviously, the absolute value of the first summand on the right-hand side is not bigger than $\eta^{-2}\Lambda_-$ and $\abs{\e_x^3} \leq \eta^{-1}\Lambda_g^2$.
We call the sum of the second and the third term on the right-hand side $\e_x^5$ and obtain $\abs{\e_x^5}\prec M^{-1/2}$ by \eqref{eq:estimate_h_xy_stoch_dom}. 
Similarly as before, we get $\abs{\e_x^2}\prec M^{-1/2}$ by using \eqref{eq:deterministic_bound_G_ij^T} instead of Lemma \ref{lem:aux_estimate_G_ij}.
An argument in the fashion of \eqref{eq:aux_est_B_x_2} yields $\abs{\e_x^4} \prec M^{-1/2}$.

Thus, by setting $\epsilon_x \defeq 2\eta^{-1} {\tilde\epsilon}^2 + \abs{\e_x^2}+ \abs{\e_x^4}+\abs{\e_x^5}$ 
and using \eqref{eq:lem_fixed_eta_aux_claim} we get 
\[ \abs{G_{x,-x}} \leq \eta^{-2}\Lambda_- +\eta^{-1}\Lambda_g^2+ \abs{\e_x^2}+ \abs{\e_x^4}+ \abs{\e_x^5} \leq \eta^{-2}\Lambda_- + 2\eta^{-3}\Lambda_-^2 +\epsilon_x. \]
Since $\epsilon_x \prec M^{-1/2}$ uniformly in $x$ the estimate \eqref{eq:inequality_Lambda-} follows from the definition $\epsilon\defeq \sup_{x} \epsilon_x$. 
\end{proof}

\subsection{Preliminary Bound on $\Lambda$} \label{sec:preliminary_bound}

In this section, we establish a deterministic bound on $\Lambda$. The proof will make essential use of the self-consistent equations in Lemma \ref{lem:self_consistent_eq}.
\begin{pro} \label{pro:fourfold_basic_bound_Lambda}
We have $\Lambda \prec M^{-\gamma/3}\Gamma^{-1}$ uniformly in $\mathbf S$.
\end{pro}

\noindent Once we have proved the two subsequent lemmas the proof of Proposition \ref{pro:fourfold_basic_bound_Lambda} follows exactly as in \cite{EJP2473}. 

\begin{lem} \label{Lem:apriori_estimate_Lambda_with_char_func}
We have the estimate $ \mathbf 1(\Lambda \leq M^{-\gamma/4}\Gamma^{-1}) \Lambda \prec M^{-\gamma/2}\Gamma^{-1}$
uniformly in $\mathbf S$.
\end{lem}

\begin{proof}
In this proof, we will use Lemma \ref{lem:aux_estimates_fourfold} (i) several times with $\varphi \defeq \mathbf 1(\Lambda \leq M^{-\gamma/4}\Gamma^{-1})$. 
Following the proof of Lemma 5.4 in \cite{EJP2473} we get \[ \varphi\Lambda_d \prec \varphi\Gamma_S\left(\Lambda^2+\sqrt\frac{\Im m +\Lambda}{M\eta}\right)\]
since $\abs{\Upsilon_x} \prec \varphi\Lambda^2 +  \sqrt{(\Im m +\Lambda)/M\eta}$ by \eqref{eq:aux_est_lambda_o_Ups_x_phi}. Moreover, because of \eqref{eq:aux_est_lambda_o_Ups_x_phi} and the first estimate 
in \eqref{eq:estimate_Gamma} we have 
\[ \varphi \Lambda_g \prec \varphi \Gamma_S\left(\Lambda^2 + \sqrt\frac{\Im m +\Lambda}{M\eta}\right).\]
Using \eqref{eq:self_averaging_G_x-x} we get
\[ \sum_{y\neq -y}  (1-m^2R_{xy})G_{y,-y} = \e_x\]
for all $x \neq -x$. Inverting $(1-m^2R)$ and using \eqref{eq:aux_est_mathcal_E_phi} yield
\begin{equation}
 \varphi \Lambda_-=\max_{x\neq -x} \varphi\abs{G_{x,-x}} \leq \Gamma_R \max_{x\neq -x} \varphi\abs{\e_x} \prec \varphi\Gamma_R \left( \Lambda^2 + \sqrt\frac{\Im m +\Lambda}{M\eta}\right).
\label{eq:appearance_Gamma_R}
\end{equation}
In total, we get 
\[ \varphi \Lambda \prec \varphi\Gamma\left( \Lambda^2 + \sqrt\frac{\Im m +\Lambda}{M\eta}\right)\]
as in (5.18) of \cite{EJP2473}. Employing the definitions of $\mathbf S$ and $\varphi$ as in the proof of Lemma 5.4 in \cite{EJP2473} establishes the claim.
\end{proof}

When estimating the off-diagonal terms $G_{x,-x}$ in \eqref{eq:appearance_Gamma_R} the control parameter $\Gamma_R$ appears naturally as the operator norm of $(1-m^2R)^{-1}$ in the same way as $\Gamma_S$ 
(which is called $\Gamma$ in \cite{EJP2473}) is used in \cite{EJP2473} to bound the differences $G_{xx}-m$. 

\begin{lem}\label{Lem:estimate_Im_z=2}
We have $\Lambda \prec M^{-1/2}$ uniformly in $z\in [-10,10]+2\mathrm i$. 
\end{lem}

\begin{proof}
We use the bounds $\abs{G_{ij}^{(\mathbb T)}}\leq 1/\eta =1/2$ from \eqref{eq:deterministic_bound_G_ij^T} and $\abs{m}\leq 1/\eta=1/2$
from the third estimate in \eqref{eq:estimate_m}. In particular, they imply $\abs{v_x}=\abs{G_{xx}-m}\leq 1$ and $\abs{m^{-1}}\geq 2$.

By \eqref{eq:inequality_Lambda-} with $\eta=2$ we have \[\Lambda_- \leq  \frac{8}{5}\epsilon \prec M^{-1/2}.\]
Thus, \eqref{eq:estimate_Lambda_g_eta_fixed} implies $\Lambda_g \prec M^{-1/2}$. Hence, $\Lambda_o \prec M^{-1/2}$ and therefore $\abs{\Upsilon_x}\prec M^{-1/2}$ by \eqref{eq:Upsilon_for_constant_eta}. 

Following now the reasoning of the proof of Lemma 5.5 in \cite{EJP2473} we get $\Lambda \prec M^{-1/2}$.
\end{proof}

\begin{proof}[Proof of Proposition \ref{pro:fourfold_basic_bound_Lambda}]
The maximum of the two Lipschitz-continuous functions $\Gamma_S$ and $\Gamma_R$ is a Lipschitz-continuous function whose Lipschitz-constant is not bigger than the maximum of the original Lipschitz-constants.
Therefore, Proposition \ref{pro:fourfold_basic_bound_Lambda} can be proved exactly in the same way as Proposition 5.3 in \cite{EJP2473}.
\end{proof}

\subsection{Proof of the Main Result}

In the whole section let $\Psi$ be a deterministic control parameter satisfying
\begin{equation}
 cM^{-1/2} \leq \Psi \leq M^{-\gamma/3}\Gamma^{-1}.
\label{eq:definition_deterministic_control}
\end{equation}

The following proposition states that such deterministic bound on $\Lambda$ can always be improved. This self-improving mechanism is also present  %TODO change formulation
in Proposition 5.6 of \cite{EJP2473}.
\begin{pro} 
\label{pro:fourfold_iteration_step}
Let $\Psi$ satisfy \eqref{eq:definition_deterministic_control} and fix $\eps \in (0,\gamma/3)$. If 
$\Lambda \prec \Psi$ %\quad \Rightarrow \quad 
then $\Lambda \prec F(\Psi)$
with $$F(\Psi)\defeq M^{-\eps}\Psi +\sqrt{\frac{\Im m}{M\eta}}+\frac{M^\eps}{M\eta}.$$
\end{pro}

\begin{proof}
We will apply the results of Lemma \ref{lem:aux_estimates_fourfold} (i) with $\varphi=1$.
Using \eqref{eq:aux_est_lambda_o_Ups_x_phi} we get
\begin{equation}
\Lambda_g + \abs{\Upsilon_x} \prec \Lambda^2+\sqrt{\frac{\Im m +\Lambda}{M\eta}} \prec \Gamma\Psi^2+\sqrt{\frac{\Im m +\Psi}{M \eta}}
\label{eq:proof_iteration_Lambda_o_Upsilon_bound}
\end{equation}
because of the first estimate in \eqref{eq:estimate_Gamma}.
The self-consistent equation \eqref{eq:self_averaging_G_x-x} for $G_{x,-x}$ implies the estimate
\begin{equation}
\abs{G_{x,-x}} \leq \abs{m^2}\left|\sum_{a\neq -a} (\E h_{xa}^2) G_{a,-a}\right| + \abs{\e_x} \prec \Gamma\Psi^2 + \sqrt{\frac{\Im m +\Psi}{M \eta}}
\end{equation}
which holds uniformly in $x$. Here, we applied the fluctuation averaging \eqref{eq:fluct_averag2_G_k-k} for $G_{x,-x}$ with $t_{xa}=\E h_{xa}^2$ and \eqref{eq:estimate_m} 
to the first summand, $\abs{\E h_{xy}^2}\leq M^{-1}$, Lemma \ref{lem:aux_estimate_G_ij} and \eqref{eq:estimate_Gamma} to the second summand and \eqref{eq:aux_est_mathcal_E_phi} to $\abs{\e_x}$ 
and employed $\Gamma_R \leq \Gamma$ and \eqref{eq:estimate_Gamma} afterwards. 

Starting with these estimates the reasoning in the proof of Proposition 5.6 in \cite{EJP2473} yields
\[\Lambda \prec \Gamma \Psi^2 + \sqrt\frac{\Im m + \Lambda}{M\eta}.\]
The claim follows from applying Young's inequality and the condition $\Psi \leq M^{-\gamma/3}\Gamma^{-1}$ to the right-hand side of the previous estimate.
\end{proof}

In the following lemma we use the notation $[v]$ for the mean of a vector $v=(v_i)_i \in \C^N$, i.e.,
$$ [v]=\frac 1 N \sum_i v_i.$$

\begin{lem} \label{lem:fluctuation_averaging_fourfold}
If $\Psi$ is a deterministic control parameter such that $\Lambda \prec \Psi$ then we have $\left[\Upsilon\right]\in O_\prec(\Psi^2)$.  
\end{lem}

\begin{proof}
%%%%%%%%%%%%%%%%%%%%%%
%%%%%% Comments contain the verification of the assumptions of Lemma \ref{lem:stoch_dom_partial_expectation}
%%%%%%%%%%%%%%%%%%%%%%
If $x \neq -x$ then we obtain from Schur's complement formula \eqref{eq:Schur_formula} and the definition of $\Upsilon_x$ 
\begin{equation}
\Upsilon_x=A_x+B_x-s_{x,-x}\E_xG_{-x,-x}^{(x)}-\E_x Y_x +\mathbb F_x\frac{1}{G_{xx}}.
\label{eq:Upsilon_x_F_xG_xx^-1}
\end{equation}
The fluctuation averaging \eqref{eq:fluct_averag1} with $t_{ik}=1/N$ yields $[\mathbb{F}_xG_{xx}^{-1}]\in O_\prec(\Psi^2)$.
Obviously, we have $\abs{A_x}\prec \Psi^2$ and $\abs{B_x}\prec \Psi^2$ by Lemma \ref{lem:aux_estimate_G_ij}. 
%Since $\mathbf S$ is a spectral domain the estimate \eqref{eq:deterministic_bound_G_ij^T} implies $\E\abs{G_{-x,-x}^{(x)}}^p \leq N^p$.  Therefore,  %%% Here: Proof that Lemma \ref{lem:stoch_dom_partial_expectation} is applicable 
Lemma \ref{lem:stoch_dom_partial_expectation}, Lemma \ref{lem:aux_estimate_G_ij} and \eqref{eq:estimate_h_xy_stoch_dom} imply 
$\abs{s_{x,-x}\E_xG_{-x,-x}^{(x)}} \prec M^{-1} \leq \Psi^2$ due to the first estimate in \eqref{eq:definition_deterministic_control}.% and \eqref{eq:Gamma_bounded_below}.

Using \eqref{eq:estimate_h_xk^2G_k-k^x-x} and the first two steps in \eqref{eq:bound_h_xk_G_kl^x-x_h_l-x_k_neq_l} with $\varphi=1$ we obtain
\begin{equation}
\left|\sum_{k,l}^{(x,-x)}h_{xk}G_{kl}^{(x,-x)}h_{l,-x}\right| \prec \Psi.
\label{eq:lem_est_fluct_averag_fourfold}
\end{equation}
Thus, the representation of $Y_x$ in \eqref{eq:expansion_Y_x} and the application of Lemma \ref{lem:aux_estimate_G_ij} yield $\abs{Y_x}\prec \Psi^2$.
Hence, Lemma \ref{lem:stoch_dom_partial_expectation} %whose requirement on the moments of $\abs{Y_x}$ is fulfilled because of \eqref{eq:finite_moments} and the second estimate in \eqref{eq:deterministic_bound_G_ij^T} 
implies $\abs{\E_xY_x} \prec \Psi^2$. For $x=-x$ the relation \eqref{eq:Upsilon_x_F_xG_xx^-1} without the second to fourth term on the right hand side and $\abs{A_x} \prec \Psi^2$ hold true   
and $\abs{[\Upsilon]}\prec \Psi^2$ follows from \eqref{eq:Upsilon_x_F_xG_xx^-1}.
\end{proof}

%\noindent After these preparations the proof of Theorem \ref{thm:Main_Fourfold} is the same as the proof of Theorem 5.1 in \cite{EJP2473}. 
\noindent Proposition \ref{pro:fourfold_basic_bound_Lambda}, Proposition \ref{pro:fourfold_iteration_step} and Lemma \ref{lem:fluctuation_averaging_fourfold} imply Theorem \ref{thm:Main_Fourfold} along the same 
lines as Proposition 5.3, Proposition 5.6 and Lemma 5.7 in \cite{EJP2473} finish the proof of Theorem 5.1 in \cite{EJP2473}.

\section{Proof of the Fluctuation Averaging} \label{sec:fluct_averag}

In this section, we verify the fluctuation averaging, i.e. Theorem \ref{thm:fluct_averag1} and Theorem \ref{thm:fluct_averag2}. To this end, we 
transfer the proof of the fluctuation averaging given in \cite{EJP2473} to our setting.
We only highlight the differences due to the special counterdiagonal terms $G_{x,-x}$.

We start with two preparatory lemmas. The following result is the analogue of Lemma B.1 in \cite{EJP2473} whose proof works in the current 
situation as well. Recall that $\E_x X = \E[X|H^{(x,-x)}]$ is the expectation conditioned on the minor $H^{(x,-x)}$ and $\F_x X = X - \E_x X$ for an integrable random variable 
$X$ (cf. definition \ref{def:minors_()} and \ref{def:partial_expectation}). 

\begin{lem} \label{lem:stoch_dom_partial_expectation}
Let $\Psi$ be a deterministic control parameter satisfying $\Psi \geq N^{-C}$ and let $X(u)$ be nonnegative random variables 
such that for every $p\in \N$ there exists a constant $c_p$ with $\E[X(u)^p]\leq N^{c_p}$ for all large $N$. If $X(u) \prec \Psi$ uniformly in $u$ then
\begin{equation*}
\E_x X(u)^n \prec \Psi^n,\qquad \mathbb F_x X(u)^n \prec \Psi^n,\qquad \E X(u)^n \prec \Psi^n
\end{equation*} 
uniformly in $u$ and in $x$.
\end{lem}

%\begin{lem} \label{lem:prep_fluct_2}
%Let $\mathbf D$ be a spectral domain, and let $\Psi_o$ and $\Psi$ be deterministic control parameters satisfying \eqref{eq:deterministic_control_general}
 %such that $\Lambda \prec \Psi$ and $\Lambda_o \prec \Psi_o$. For a fixed finite subset $\mathbb T \subset \N$, $i \neq j$ and $i, j \notin \mathbb T$ 
%we have 
%\begin{equation}
%\abs{G_{ij}^{(\mathbb T)}}\prec \Psi_o,\quad \bigg|\frac{1}{G_{ii}^{(\mathbb T)}}\bigg|\prec 1.
%\label{eq:bounds_G_ij_fluct}
%\end{equation}
%Furthermore, $\abs{G_{ij}^{(\mathbb T)}}\leq M$ and for every $n\in \N$ and $\eps>0$ there is $N_0\in \N$ such that
%\begin{equation}
%\E \abs{G_{ii}^{(\mathbb T)}}^{-n} \leq N^\eps
%\label{eq:estimate_E_abs_G_ii^-n}
%\end{equation}
%for all $N\geq N_0$.
%\end{lem}

This Lemma will be used throughout the following arguments. The trivial condition $\E[X(u)^p]\leq N^{c_p}$ will always be fulfilled. 
%
%Lemma B.2 in \cite{EJP2473} holds in the verbatim formulation also in the case of the fourfold symmetry. 
%Since we do not need it here explicitely we omit its formulation and proof.  
The following Lemma which replaces (B.5) in \cite{EJP2473} gives an auxiliary bound for estimating high moments of $\abs{\sum_k t_{ik}\F_k G_{kk}^{-1}}$
when there are bounds on $\Lambda = \max_{x,y} \abs{G_{xy}-\delta_{xy}m}$ and $\Lambda_o=\max_{x\neq y} \abs{G_{xy}}$ (cf. \eqref{eq:definition_Lambda}). 

\begin{lem} \label{lem:prep_fluct_3}
Let $\mathbf D$ be a spectral domain.
Suppose $\Lambda \prec \Psi$ and $\Lambda_o \prec \Psi_o$ for some deterministic control parameters $\Psi$ and $\Psi_o$ which 
satisfy \eqref{eq:deterministic_control_general}. Then for fixed $p \in \N$ we have 
\begin{equation}
\left| \mathbb F_x \left(G_{xx}^{(\mathbb T)}\right)^{-1} \right| \prec \Psi_o
\label{eq:prep_fluct3}
\end{equation}
uniformly in $\mathbb T \subset \N$, $\abs{\mathbb T} \leq p$, $x \notin \mathbb T \cup -\mathbb T$ and $z \in \mathbf D$.
\end{lem}

\begin{proof}
If $x=-x$ then the proof of \eqref{eq:prep_fluct3} is exactly the same as the proof of (B.5) in \cite{EJP2473}. 
For $x\neq -x$ we start with \eqref{eq:Schur_formula}. 
Since $x,-x \notin \mathbb T$ we obtain as in the proof of \eqref{eq:expansion_expectation_Schur} by using the first resolvent identity 
\eqref{eq:resolvent_identity1} that 
\begin{align}
\sum_{a,b}^{(\mathbb T, x)} h_{xa}G_{ab}^{(\mathbb T, x)}h_{bx}=& C_x^{(\mathbb T)}
%\nonumber\\ &
+\sum_{a,b}^{(\mathbb T,x,-x)} h_{xa}G_{ab}^{(\mathbb T,x,-x)}h_{bx} + \left(G_{-x,-x}^{(\mathbb T,x)}\right)^{-1}\sum_{a,b}^{(\mathbb T,x,-x)}
h_{xa}G_{a,-x}^{(\mathbb T,x)}G_{-x,b}^{(\mathbb T,x)}h_{bx},
\label{eq:exp_Schur_last_summand_T}
\end{align}
where we used the definition
$$C_x^{(\mathbb T)}\defeq h_{x,-x}G_{-x,-x}^{(\mathbb T,x)}h_{-x,x}+\sum_{a}^{(\mathbb T,x,-x)}h_{xa}G_{a,-x}^{(\mathbb T,x)}h_{-x,x}+
\sum_{b}^{(\mathbb T,x,-x)} h_{x,-x}G_{-x,b}^{(\mathbb T,x)}h_{bx}.$$
The assumptions of Lemma \ref{lem:stoch_dom_partial_expectation} are fulfilled for each term of the expansion in 
\eqref{eq:exp_Schur_last_summand_T} by \eqref{eq:estimate_h_xy_stoch_dom} and the second estimate in \eqref{eq:deterministic_bound_G_ij^T}. 

Similar to the proof of \eqref{eq:est_C_x} we get $\abs{C_x^{(\mathbb T)}}\prec M^{-1/2}\leq \Psi_o$ by \eqref{eq:deterministic_control_general}.
Using the first step in \eqref{eq:first_estimate_X_x} and the argument in \eqref{eq:est_sum_abs_h^2-s_G_ii} we get 
\begin{equation*}
\left|\mathbb F_x\sum_{a,b}^{(\mathbb T,x,-x)} h_{xa}G_{ab}^{(\mathbb T,x,-x)}h_{bx}\right| \leq  \left|\sum_{a\neq b}^{(\mathbb T,x,-x)} h_{xa}
G_{ab}^{(\mathbb T, x, -x)}h_{bx}\right|+\left|\sum_{a}^{(\mathbb T,x,-x)}\left(\abs{h_{xa}}^2-s_{xa}\right)G_{aa}^{(\mathbb T, x, -x)}\right| 
\prec \Psi_o
\end{equation*}
where we used that $\Psi_o$ fulfills \eqref{eq:deterministic_control_general}.
The estimate 
\begin{equation}
\left|\sum_{k,l}^{(\mathbb T,x,-x)} h_{xk}G_{kl}^{(\mathbb T,x,-x)}h_{l,-x}\right| \prec \Psi_o
\label{eq:lem_fluct_averag_est_Eh^2=0}
\end{equation}
which follows from adapting \eqref{eq:estimate_h_xk^2G_k-k^x-x} and the first step in \eqref{eq:bound_h_xk_G_kl^x-x_h_l-x_k_neq_l} implies
%Hence, as $\abs{G_{-x,-x}^{(\mathbb Tx)}}\prec 1$ by Lemma \ref{lem:aux_estimate_G_ij} we get 
\begin{equation}
\left|\left(G_{-x,-x}^{(\mathbb T,x)}\right)^{-1}\sum_{a,b}^{(\mathbb T,x,-x)}
h_{xa}G_{a,-x}^{(\mathbb T,x)}G_{-x,b}^{(\mathbb T,x)}h_{bx}\right| \prec \Psi_o^2 \prec \Psi_o
\label{eq:estimate_fluct_averag_Y_x_general}
\end{equation}
using a similar representation as in \eqref{eq:expansion_Y_x} and Lemma \ref{lem:aux_estimate_G_ij}.
By Lemma \ref{lem:stoch_dom_partial_expectation} these estimates imply $$\left|\mathbb F_x  \sum_{a,b}^{(\mathbb T, x)} 
h_{xa}G_{ab}^{(\mathbb T, x)}h_{bx}\right|\prec \Psi_o.$$
Thus, the claim is obtained by applying Schur's complement formula \eqref{eq:Schur_formula} to $G_{xx}^{(\mathbb T)}$ 
and observing that $\abs{\mathbb F_x(h_{xx}-z)}=\abs{h_{xx}}\prec M^{-1/2} \leq \Psi_o$ as $h_{xx}$ is independent of $H^{(x,-x)}$ and $\E h_{xx}=0$.
\end{proof}

%%%%%%%%%%%%%%%%%%%%%%%%%%%%%%%%%%%%%%%%%%%%%%%%%%%%%%%%%%%%%%%%%%%%%%%%%%%%%%%
%%%       Proof of the strong version of the Fluctuation Averaging          %%%
%%%%%%%%%%%%%%%%%%%%%%%%%%%%%%%%%%%%%%%%%%%%%%%%%%%%%%%%%%%%%%%%%%%%%%%%%%%%%%%
\begin{proof}[Proof of Theorem \ref{thm:fluct_averag2}] 
The proof is similar to the proof of Theorem 4.7 on pages 48 to 53 in \cite{EJP2473} so we only describe the changes needed to transfer this proof %of Theorem 4.7 on pages 48 to 53 in \cite{EJP2473} 
to its version for the fourfold symmetry. 

First, we use Lemma \ref{lem:prep_fluct_3} instead of (B.5). Moreover, we have to change some notions introduced in 
the proof of Theorem 4.7. In the middle of page 49, an equivalence relation on the set $\{1,\ldots,p\}$ is introduced which has to be substituted
by the following equivalence relation. Starting with $\mathbf k\defeq (k_1,\ldots,k_p)\in (\Z\slash N\Z)^p$ and $r,s \in \{1,\ldots,p\}$ we
define $r \sim s$ if and only if $k_r =k_s$ or $k_r=-k_s$. As in \cite{EJP2473} the summation over all $\mathbf k$ is regrouped with respect 
to this equivalence relation and the notion of ``lone'' labels has to be understood with respect to this equivalence relation. We use the same notation
$\mathbf k_L$ for the set of summation indices corresponding to lone labels. Differing from the definition in \cite{EJP2473} we call a
resolvent entry $G_{xy}^{(\mathbb T)}$ with $x,y \notin \mathbb T$ \emph{maximally expanded} if $\mathbf k_L \cup - \mathbf k_L \subset 
\mathbb T \cup \{x,y\}$. Correspondingly, we denote by $\mathcal A$ the set of monomials in the off-diagonal entries $G_{xy}^{(\mathbb T)}$ with 
$\mathbb T \subset \mathbf k_L \cup -\mathbf k_L$, $x\neq y$ and $x,y \in \mathbf k \backslash \mathbb T$ (considering $\mathbf k$ as a subset of 
$\Z\slash N\Z$) and the inverses of diagonal entries $1/G_{xx}^{(\mathbb T)}$ with $\mathbb T \subset \mathbf k_L \cup -\mathbf k_L$ and 
$x \in \mathbf k \backslash \mathbb T$. With these alterations the algorithm can be applied as in \cite{EJP2473}. In the proof of (B.15) 
the assertion $(*)$ has to be replaced by  

\vspace*{0.1cm}
\begin{tabular}{cl}
$(*)$ & For each $s \in L$ there exists $r=\tau(s) \in \{1,\ldots,p\}\backslash\{s\}$ such that the monomial $A^r_{\sigma_r}$\\ 
&contains a resolvent entry with lower index $k_s$ or $-k_s$.  
\end{tabular}
\vspace*{0.1cm}

To prove this claim, we suppose by contradiction that there is $s \in L$ such that $A^r_{\sigma_r}$ does not contain $k_s$ and $-k_s$ as
lower index for all $r\in \{1,\ldots,p\}\backslash\{s\}$. Without loss of generality we assume $s=1$. This implies that each resolvent entry
in $A_{\sigma_r}^r$ contains $k_1$ and $-k_1$ as upper index since $A_{\sigma_r}^r$ is maximally expanded for all $r\in\{2,\ldots,p\}$.
Therefore, $A_{\sigma_r}^r$ is independent of $k_1$ as defined in Definition \ref{def:independent}. Using \eqref{eq:independent_expectation}
and proceeding as in \cite{EJP2473} concludes the proof of  $(*)$. 

Following verbatim the remaining steps in the proof of Theorem 4.7 in \cite{EJP2473} establishes the assertion of Theorem \ref{thm:fluct_averag2}.
\end{proof}

\noindent Now, we deduce Theorem \ref{thm:fluct_averag1} from Theorem \ref{thm:fluct_averag2}. %As we use \eqref{eq:lem_fluct_averag_est_Eh^2=0} of 
%\eqref{eq:estimate_Y_x_fluct_averag} in the following proof of \eqref{eq:fluct_averag2} this result also makes use of the assumption $\E h_{xy}^2=0$.
\begin{proof}[Proof of Theorem \ref{thm:fluct_averag1}]
The first estimate in \eqref{eq:fluct_averag1} follows from Theorem \ref{thm:fluct_averag2} directly by setting $\Psi_o\defeq \Psi$ and using $\Lambda_o \leq \Lambda \prec \Psi_o$.

To verify the second estimate in \eqref{eq:fluct_averag1} we use the fourth estimate in Lemma \ref{lem:aux_estimate_G_ij} which implies 
\begin{equation}
\abs{\F_x G_{xx}^{(\mathbb T)}} =\abs{\F_x \left(G_{xx}^{(\mathbb T)}-m\right)} \prec \Psi.
\label{eq:est_F_x_G_xx_T}
\end{equation}
%by Lemma \ref{lem:stoch_dom_partial_expectation} since $\abs{G_{xx}^{(\mathbb T)}-m} \leq M+1$.
Now, following the proof of Theorem \ref{thm:fluct_averag2} verbatim with $\Psi_o\defeq \Psi$ and replacing the usage of Lemma \ref{lem:prep_fluct_3}
by \eqref{eq:est_F_x_G_xx_T} yield the second estimate in \eqref{eq:fluct_averag1}.

Similarly, the third estimate in \eqref{eq:fluct_averag1} is proved by following the proof of Theorem \ref{thm:fluct_averag2} verbatim with $\Psi_o\defeq \Psi$ and 
Lemma \ref{lem:prep_fluct_3} replaced by 
\[ \abs{\F_x G_{x,-x}^{(\mathbb T)}} \prec \Lambda_o \prec \Psi\]
for $x\neq -x$ which is a consequence of Lemma \ref{lem:aux_estimate_G_ij} and Lemma \ref{lem:stoch_dom_partial_expectation}.

Next, we establish \eqref{eq:fluct_averag2}. We start from Schur's complement formula \eqref{eq:Schur_formula} with $\mathbb T=\emptyset$ and use \eqref{eq:semicircle2} to get
\begin{equation}
\frac{1}{G_{xx}}=\frac{1}{m}+h_{xx}-\left(\sum_{k,l}^{(x)}h_{xk}G_{kl}^{(x)}h_{lx}-m\right).
\label{eq:Schur_manipulated}
\end{equation}
Using Lemma \ref{lem:aux_estimate_G_ij} with $\varphi=1$ and the first estimate in \eqref{eq:estimate_m} we get 
$$\left|\frac{1}{G_{xx}}-\frac{1}{m}\right|=\left|\frac{G_{xx}-m}{G_{xx}m}\right|\prec \abs{G_{xx}-m}\prec \Psi.$$
Thus, $\abs{h_{xx}-\left(\sum_{k,l}^{(x)}h_{xk}G_{kl}^{(x)}h_{lx}-m\right)}\prec \Psi$ as well. Therefore, we can expand the inverse of the right-hand 
side of \eqref{eq:Schur_manipulated} around $1/m$ which yields 
\begin{equation}
v_x\defeq G_{xx}-m=m^2\left(-h_{xx}+\sum_{k,l}^{(x)}h_{xk}G_{kl}^{(x)}h_{lx}-m\right) +g_x
\label{eq:exp_v_x_fluct}
\end{equation}
with error terms $g_x$ such that $\abs{g_x}\prec \Psi^2$ uniformly in $x$.
By \eqref{eq:expansion_expectation_Schur}, \eqref{eq:expansion_term4}, \eqref{eq:definition_C_x} and \eqref{eq:definition_Y_x} we have for $x\neq -x$ the representation
\begin{equation}
\sum_{k,l}^{(x)}h_{xk}G_{kl}^{(x)}h_{lx}=\sum_a s_{xa} G_{aa}-A_x-B_x-s_{-x,x}G_{-x,-x}^{(x)}+Z_x+Y_x+C_x+s_{-x,x}G_{-x,-x}^{(x)}.
\label{eq:estimate_fluct_aux_2}
\end{equation}
Taking the expectation $\E_x$ of \eqref{eq:exp_v_x_fluct} we want to prove that 
\begin{equation}
\E_x v_x = m^2 \sum_a s_{xa} v_a +f_x,
\label{eq:representation_E_x_v_x}
\end{equation}
where $\abs{f_x}\prec \Psi^2$ uniformly in $x$. From \eqref{eq:expansion_term4} we get that the sum of the first four summands on the right-hand side 
of \eqref{eq:estimate_fluct_aux_2} 
is $H^{(x,-x)}$-measurable. Therefore, it suffices to show that all summands except the first one on the right-hand 
side of \eqref{eq:estimate_fluct_aux_2} are bounded by $\Psi^2$ uniformly in $x$. For $A_x$ and $B_x$ this follows directly from their definitions 
in \eqref{eq:definition_A_x_B_x}. Since $Z_x=\F_x X_x$ for some random variable $X_x$ we get $\E_x Z_x=0$. The representation 
\eqref{eq:second_representation_C_x} for $C_x$ and Lemma \ref{lem:aux_estimate_G_ij} yield
$\abs{C_x} \prec M^{-1} +M^{-1/2}\Psi \prec \Psi^2$
by \eqref{eq:deterministic_control_general}.  The bound \eqref{eq:estimate_fluct_averag_Y_x_general} with $\mathbb T=\emptyset$ gives $\abs{Y_x}\prec \Psi^2$
uniformly in $x$. 
If $x=-x$ then the argumentation in \cite{EJP2473} can be applied. 
This finishes the proof of \eqref{eq:representation_E_x_v_x}. 

Therefore, since $\E_x+\F_x=1$ we have
\begin{align}
w_a \defeq \sum_x t_{ax}v_x& =\sum_x t_{ax}\E_x v_x +\sum_x t_{ax}\F_x v_x = m^2\sum_{x,y} t_{ax}s_{xy}v_y + F_a \nonumber\\
&=m^2 \sum_{x,y}s_{ax}t_{xy} v_y +F_a = m^2 \sum_x s_{ax} w_x +F_a \label{eq:representation_w_a},
\end{align}
where we used \eqref{eq:representation_E_x_v_x} with the notation $F_a \defeq \sum_x t_{ax}(f_x + \F_x v_x)$ in the third step and in the fourth step that $T$ and $S$ commute.
Note that $\abs{F_a} \prec \Psi^2$ uniformly in $a$ as $\abs{\sum_x t_{ax}\F_x v_x}=\abs{\sum_x t_{ax}\F_xG_{xx}} \prec \Psi^2$ 
by the second estimate in \eqref{eq:fluct_averag1}. Introducing the vectors $\mathbf w \defeq (w_a)_{a\in \Z\slash N\Z}$ and $\mathbf F\defeq 
(F_a)_{a\in \Z\slash N\Z}$ and writing \eqref{eq:representation_w_a} in matrix form we get
$$\mathbf w= m^2 S\mathbf w +\mathbf F.$$
Inverting the last equation yields
$$\mathbf w = (1-m^2S)^{-1} \mathbf F.$$
Recalling the definition \eqref{eq:definition_Gamma} we have
$$\norm{\mathbf w}_\infty \leq \Gamma_S \norm{\mathbf F}_\infty \prec \Gamma_S \Psi^2$$
since $\abs{F_a}\prec \Psi^2$ uniformly in $a$ is equivalent to $\norm{\mathbf F}_\infty \prec\Psi^2$. This proves \eqref{eq:fluct_averag2}.

%The proof of \eqref{eq:fluct_averag2_G_k-k} follows the same idea as the previous proof of \eqref{eq:fluct_averag2}. 
In order to prove \eqref{eq:fluct_averag2_G_k-k} it suffices to verify 
\begin{equation}
\E_x G_{x,-x} = m^2\sum_{a\neq -a} (\E h_{xa}^2) G_{a,-a} + f_x 
\label{eq:representation_E_x_G_x-x}
\end{equation}
with $\abs{f_x}\prec \Psi^2$ uniformly in $x$. Then \eqref{eq:fluct_averag2_G_k-k} follows from the same reasoning as in the proof of \eqref{eq:fluct_averag2} with $S$ replaced by $R$ and 
\[w_x\defeq \sum_{a\neq -a}t_{xa}G_{a,-a}.\] 
To compute the partial expectation $\E_x G_{x,-x}$ we use the expansion 
\begin{align}
G_{x,-x}= & ~~m^2 \sum_a^{(x,-x)}(\E h^2_{xa})G_{a,-a}^{(x,-x)} + m^2 \sum_{a\neq b}^{(x,-x)}h_{xa}G_{a,-b}^{(x,-x)}h_{xb} +m^2 \sum_a^{(x,-x)}\left(h_{xa}^2-\E h^2_{xa}\right)G_{a,-a}^{(x,-x)} \nonumber \\
 & + (m^2-G_{xx}G_{-x,-x}^{(x)})h_{x,-x} -m^2 h_{x,-x} +(G_{xx}G_{-x,-x}^{(x)}-m^2)\sum_{a\neq b}^{(x,-x)}h_{xa}G_{a,-b}^{(x,-x)}h_{xb} \nonumber \\
 & +(G_{xx}G_{-x,-x}^{(x)}-m^2)\sum_a^{(x,-x)}h^2_{xa}G_{a,-a}^{(x,-x)}
\label{eq:expansion_G_x-x_fluct_averag}
\end{align}
which follows from the resolvent identities in a similar way as \eqref{eq:self_averaging_G_x-x}.

The first summand in \eqref{eq:expansion_G_x-x_fluct_averag} is $H^{(x,-x)}$-measurable. Using \eqref{eq:resolvent_identity1} twice and adding the two missing terms we obtain the first 
summand on the right-hand side of \eqref{eq:representation_E_x_G_x-x}. The error terms originating from the usage of the resolvent identities and the added terms are obviously dominated by $\Psi^2$.
The partial expectations with respect to $H^{(x,-x)}$ of the second and the fifth term vanish. 
For the remaining terms we use Lemma \ref{lem:stoch_dom_partial_expectation}.  First, 
$\abs{m^2-G_{xx}G_{-x,-x}^{(x)}}\prec \Psi$ because of the triangle inequality, Lemma \ref{lem:aux_estimate_G_ij} and the second estimate in \eqref{eq:estimate_m}.
Thus, using \eqref{eq:estimate_h_xy_stoch_dom} and \eqref{eq:def_determ_con_para} for the fourth term, the first step in \eqref{eq:bound_h_xk_G_kl^x-x_h_l-x_k_neq_l} for the sixth term and \eqref{eq:estimate_h_xk^2G_k-k^x-x} 
for the seventh term 
we get that these summands are dominated by $\Psi^2$. Similarly to \eqref{eq:est_sum_abs_h^2-s_G_ii} we see that the third summand is dominated by $\Psi^2$ using the Large Deviation Bound (C.2) in 
\cite{EJP2473} and the first estimate in Lemma \ref{lem:aux_estimate_G_ij}. Lemma \ref{lem:stoch_dom_partial_expectation} establishes \eqref{eq:representation_E_x_G_x-x} which finishes the proof of 
Theorem \ref{thm:fluct_averag1}.
\end{proof}

\bibliographystyle{amsplain}
\bibliography{literatur}

\end{document}